\documentclass[12pt,a4paper]{article}
\usepackage{setspace}
\usepackage[utf8]{inputenc}
\usepackage{amsmath}
\usepackage{amsfonts}
\usepackage{amssymb}
\usepackage{graphicx}
\usepackage{placeins}
\usepackage{array}
\usepackage{amsthm}
\usepackage{caption}
\usepackage{subcaption}
\usepackage{appendix}
\usepackage{multirow}
\usepackage{booktabs}
\usepackage{verbatim}
\usepackage{natbib}

\theoremstyle{remark}
\newtheorem{remark}{Remark}

\theoremstyle{definition}
\newtheorem{theorem}{Theorem}

\theoremstyle{definition}

\usepackage[left=1in,right=1in,top=1in,bottom=1in]{geometry}

\DeclareMathOperator{\tr}{tr}
\DeclareMathOperator{\diag}{diag}
\DeclareMathOperator{\var}{Var}

\DeclareMathOperator{\vect}{vec}
\DeclareMathOperator{\rank}{rank}

\usepackage{hyperref}
\hypersetup{
    colorlinks=true, 
    linktoc=all,     
    linkcolor=blue,  
    citecolor=cyan
}

\title{A Bayesian Generalized Bridge Regression Approach to Covariance Estimation in the Presence of Covariates}
\author{Christina Zhao, Ding Xiang, Galin L. Jones, and Adam J. Rothman}
\date{\today}

\begin{document}

\maketitle

\begin{abstract}
A hierarchical Bayesian approach that permits simultaneous inference for the regression coefficient matrix and the error precision (inverse covariance) matrix in the multivariate linear model is proposed.  Assuming a natural ordering
of the elements of the response, the precision matrix is reparameterized so it can be estimated with univariate-response linear regression techniques. A novel generalized bridge regression prior that accommodates both sparse and dense settings and is competitive with alternative methods for univariate-response regression is proposed and used in this framework. Two component-wise Markov chain Monte Carlo algorithms are developed for sampling, including a data augmentation algorithm based on a scale mixture of normals representation. Numerical examples demonstrate that the proposed method is competitive with comparable joint mean-covariance models, particularly in estimation of the precision matrix. 
The method is also used to estimate the $253 \times 253$ precision matrix of 90,670 spectra extracted from images taken by the Hubble Space Telescope, demonstrating its computational feasibility for problems with large $n$ and $q$. 
\end{abstract}

\section{Introduction}\label{sec:intro}
The simultaneous modeling of multiple numerical response variables is a fundamental problem. 
Examples include predicting infrared spectra from chemical structures \citep{saqueretal2024} and modeling gravitational waves \citep{engelsetal2014}. 
A Bayesian shrinkage estimation framework is proposed under the following model. Let $Y_i$ be a $q$-variate response random vector and $X_i\in\mathbb{R}^p$ be the associated covariate vector for the $i$th subject, and let $B\in\mathbb{R}^{p\times q}$ be the regression coefficient matrix. Define $\varepsilon_i$ to be a $q$-variate latent random vector with zero mean and positive definite, diagonal covariance matrix $D$ and $L\in\mathbb{R}^{q\times q}$ to be a lower triangular matrix with ones on its diagonal. The distribution of $Y_i$ is described by the multivariate linear model
\begin{equation}\label{eq:multimeancov}
Y_i = B^\top X_i + L \varepsilon_i, \quad i=1,\ldots, n.
\end{equation}
Denote the covariance matrix of the $Y_i$'s by $\Omega^{-1}$.

Simultaneous estimation of $(B,\Omega)$ is desirable, as accounting for correlation between the responses can improve prediction \citep{breimanfriedman1997} and improve shrinkage estimation of $B$ \citep{rothmanetal2010multireg}. Two major challenges in this problem are the number of parameters, which increase linearly in $p$ and quadratically in $q$, and the positive definite constraint on $\Omega$. To address the dimensionality of the problem, many methods impose sparsity on $B$ and $\Omega$, often through extending methods for sparse univariate-response linear regression. Several frequentist methods use a penalized likelihood approach with the $l_1$ penalty imposed on both $B$ and $\Omega$ \citep{rothmanetal2010multireg, lee12, caietal2013}. Bayesian approaches include the multivariate spike-and-slab LASSO (mSSL) \citep{deshpandeetal2019} and the horseshoe-graphical horseshoe (HS-GHS) \citep{lietal2021}, which specify priors designed for sparse univariate-response regression in an element-wise fashion on $B$ and $\Omega$. \cite{bhadramallick2013} and \cite{bottoloetal2021} both specify spike-and-slab priors on $B$ and a hyper inverse Wishart distribution on $\Omega$, but the algorithms focus on structure learning and only estimate the sparsity structure of $\Omega$.

In contrast to the previously mentioned approaches, which enforce positive-definiteness of $\Omega$ by constraining the prior, the hierarchical model proposed here guarantees positive-definiteness by construction. This is achieved through a reparameterization of $\Omega$.
Define $T = L^{-1}$, which is also lower triangular with ones on its diagonal. 
The modified Cholesky decomposition is given by 
\begin{equation}\label{eq:modcholprec}
\Omega=T^\top D^{-1}T.
\end{equation}
The parameters in $T$ are not subject to constraint, and as long as the diagonal elements of $D$ are positive, $\Omega$ is guaranteed to be positive definite and symmetric. \cite{pourahmadi1999} showed that the parameters of $T$ and $D$ can be estimated through a sequence of $q$ autoregressions. As a consequence, the resulting estimator of $\Omega$ is not invariant to permutations of the variables and thus is more suited to applications where $Y_i$ has a natural ordering, such as longitudinal or spatial data \citep{wu03, kiddkatzfuss2021}, though it has been applied to financial returns data for computing the covariance matrix for $q$ assets \citep{carvalhoetal2010,gramacypantaleo2010}. For applications where a natural ordering is not known, a method such as the Isomap algorithm \citep{wagamanlevina2009} may be used to first find a structured ordering, or a final estimate for $\Omega$ can also be constructed by pooling estimates from multiple permutations \citep{kangdeng2020,zhengetal2017}. However, a natural ordering of the variables is assumed to be available here.

A variety of methods in the covariance estimation literature have leveraged the modified Cholesky decomposition to extend univariate-response regression priors to estimating $\Omega$. \cite{danielspourahmadi2002} used conjugate priors without making sparsity assumptions, while many others assume a banded structure \citep{kiddkatzfuss2021,leelee2021,leelin2023} or an arbitrary sparsity pattern for $T$ \citep{smithkohn2002,leeetal2019} as part of prior specification. \cite{gramacypantaleo2010} consider ridge and lasso priors, but penalized regression approaches are often considered under a frequentist paradigm \citep{huangetal2006,levinaetal2008}. Many of these methods focus only on covariance estimation and estimate $\mathbb{E}(Y_i)$ with the sample mean; among the methods mentioned here, only \cite{danielspourahmadi2002} and \cite{smithkohn2002} also consider mean estimation in the presence of covariates.

Although a variety of regression priors have been considered, practitioners rarely know a priori whether they are in a dense or sparse setting, and sparse methods are not always optimal; in univariate-response regression, it is well established that lasso and ridge regression perform best in different settings \citep{tibshirani1996,fu1998}. A novel penalized regression prior that addresses this by incorporating both sparse and dense settings in a single prior is proposed here. Whereas use of the ridge and lasso regression penalties would require the practitioner to consider penalty selection in model fitting, the proposed prior addresses this nuisance parameter using a type of model averaging approach. Furthermore, the range of the penalty is extended beyond the commonly used range of $[1,2]$. Allowing penalty values less than 1 encourages sparsity when appropriate, while values greater than 2 improve performance in dense settings. The proposed generalized bridge (GBR) prior also uses local shrinkage parameters. This results in a local-global interpretation and leads to the same sort of tail-robustness properties enjoyed by the horseshoe \citep{carvalhoetal2010}. In univariate-response estimation and prediction, the GBR prior is competitive with the horseshoe and the spike-and-slab lasso \citep{rockovageorge2018}. 

Extending the GBR prior to the multivariate linear regression model in equation \eqref{eq:multimeancov} is conceptually straightforward. When $\Omega$ is reparameterized using the modified Cholesky decomposition, simultaneous estimation of $(B,\Omega)$ is converted to estimation of $q$ univariate-response autoregressions. However, efficient computation is more difficult due to posterior dependence of $B$ and the Cholesky factors of $\Omega$. 
Two Markov chain Monte Carlo (MCMC) algorithms are developed for full posterior inference. The first is a standard component-wise algorithm that can accommodate values of the penalty parameter greater than 2. However, it requires $pq + q(q-1)/2 + 2$ Metropolis-Hastings updates, making it difficult to tune in high dimensions. The second algorithm addresses this using data augmentation, and although unable to allow penalty values greater than 2, it is computationally more efficient, reducing the number of required Metropolis-Hastings updates to 2. 

The computational efficiency of the data augmentation algorithm lies in its use of 
a scale mixture of normals (SMN) representation of the exponential power distribution \citep{west1987} to augment the posterior conditionals of the elements of $B$ and the lower triangular matrix $T$. This allows sampling using multivariate normal distributions, even when $p > n$ and $q > n$, without changing the posterior conditionals of the other parameters. While standard calculations suffice for demonstrating this for $T$, sampling from a multivariate normal distribution for $B$ in the $p > n$ case requires a transformation that uses $T$ and the $p\times p$ orthonormal matrix of the singular value decomposition (SVD) of the covariate matrix $X$. 

When $n$ is large relative to $p$ and $q$, other computational considerations in the development of the data augmentation algorithm give the proposed model distinct advantages over the available implementation of HS-GHS, which also allows full posterior inference. The data augmentation algorithm uses two sampling strategies for the multivariate normal distribution: the approach of \cite{bhattacharyaetal2016} is used when the dimension of the normal distribution to be sampled is greater than the sample size, and the approach of \cite{rue2001} is used otherwise. As a result, computation time is much lower compared to the HS-GHS algorithm when $p \ll n$ (Section \ref{sec:compeffort}). Additionally, when $p < n$, products involving matrices with $n$ rows are pre-computed, which allows the proposed model to be fitted to datasets with large $n$ that HS-GHS is unable to fit due to memory constraints.
The data example in Section \ref{sec:wisp} uses spectra extracted from images obtained by the Hubble Space Telescope. These spectra have length $q=253$ and are divided into two classes based on the presence or absence of emission lines. The training sample for the spectra without emission lines had a sample size of $n=90,670$. In the absence of covariates, the HS-GHS algorithm requires inversion of an $nq\times nq$ matrix to sample the mean structure, and for this data example, 128 gigabytes (GB) of random-access memory (RAM) was insufficient for completing this operation for a single iteration of the algorithm. On the other hand, because $p=1 < n$ and $q < n$, the pre-computation done by the proposed data augmentation algorithm results in the largest matrix required being $q\times q$. Consequently, 16 GB of RAM was sufficient for 5000 iterations of the data augmentation algorithm in this example.

The remainder proceeds as follows. Section \ref{sec:generalizedbridge} develops the GBR prior for univariate-response regression, and Section \ref{sec:gmcb} extends it to estimation of $(B,\Omega)$ to construct the generalized mean-covariance bridge (GMCB) prior. Point estimates and the sampling algorithms are discussed in Section \ref{sec:est}. Estimation performance is compared to other joint mean-covariance methods in Section \ref{sec:sim}, and GMCB is demonstrated on emission spectra in Section \ref{sec:wisp}. Concluding remarks are offered in Section \ref{sec:finalremarks}.

\section{The Generalized Bridge Prior}\label{sec:generalizedbridge}
The modified Cholesky decomposition transforms the estimation of $\Omega$ into a sequence of univariate-response regression problems. A penalized regression approach that allows elimination of the nuisance parameters through the use of marginal densities for inference is developed in this setting, before considering its extension to the multivariate linear model in Section \ref{sec:gmcb}.

Many penalized regression methods can be described by a common framework. Let $\lVert\cdot\rVert_2$ denote the Euclidean norm. If $Y$ is an $n\times 1$ vector of centered responses, $X$ is an $n\times p$ standardized matrix of covariates, and $\beta$ is a $p\times 1$ vector of regression coefficients, for fixed $\alpha>0$ and $\lambda > 0$, the frequentist penalized regression estimate is the solution to
\begin{align}
    \arg\min_\beta\lVert Y-X\beta\rVert_2^2 + \lambda\sum_{j=1}^p|\beta_j|^\alpha. \label{eq:frequentistbridge}
\end{align}
Choices of $\alpha = 1$ and $\alpha = 2$ correspond to the frequentist lasso \citep{tibshirani1996} and ridge \citep{hoerlkennard1970} estimates, respectively. Values of $\alpha$ other than $1$ and $2$ in $(0,2]$ correspond to the frequentist bridge estimate \citep{frankfriedman1993}, though values of $\alpha < 1$ have been limited in application due to non-convexity.

Bayesian penalized regression priors may be constructed so that the frequentist estimates are the posterior modes. If $I_n$ denotes the $n\times n$ identity matrix, a standard Bayesian formulation of penalized regression models assumes
\begin{align}
    Y|X,\beta,\sigma^2 & \sim N(X\beta,\sigma^2I_n), \nonumber \\
    \nu(\beta|\lambda,\alpha,\sigma^2) &= \left(\frac{\alpha \lambda^{1/\alpha}}{2^{1/\alpha+1}(\sigma^2)^{1/\alpha}\Gamma(1/\alpha)}\right)^p\exp\left\lbrace-\frac{\lambda}{2\sigma^2}\sum_{j=1}^p|\beta_j|^\alpha\right\rbrace. \label{eq:standardbayesbridge}
\end{align}
For fixed $(\lambda,\alpha)$ and $\nu(\sigma^2) \propto 1/\sigma^{2}$, the marginal posterior distribution of $\beta$ is characterized by
\begin{align}
    q(\beta|Y)\propto\left[\lVert Y-X\beta\rVert_2^2 + \lambda\sum_{j=1}^p|\beta_j|^\alpha\right]^{-\left(\frac{n}{2}+\frac{p}{\alpha}+1\right)}, \label{eq:standardbridgepost}
\end{align}
and hence the mode of this distribution is the solution to equation \eqref{eq:frequentistbridge}.

Equation \eqref{eq:standardbridgepost} requires a choice of $\lambda$ and $\alpha$. While $\lambda$ is often modeled with a prior or selected through methods such as cross-validation or empirical Bayes approaches, similar treatment of $\alpha$ is not widespread. \cite{polsonetal2014} consider a prior for $\alpha\in(0,1)$, but other approaches fix $\alpha$ at a pre-selected or estimated value \citep{parkcasella2008,mallickyi2018,armagan2009,griffinhoff2020}, despite the optimal choice of $\alpha$ varying based on the nature of the unknown $\beta$ \citep{tibshirani1996,fu1998}. As shown in Figure \ref{fig:samplebetapriors}, smaller values of $\alpha$ accommodate large signals and sparsity, while larger values accommodate small non-zero signals. For estimating an unstructured covariance matrix, it is necessary to accommodate both sparse and dense settings, which is difficult when $\alpha$ is a fixed value. 

\begin{figure}
    \centering
    \includegraphics[width=0.9\textwidth]{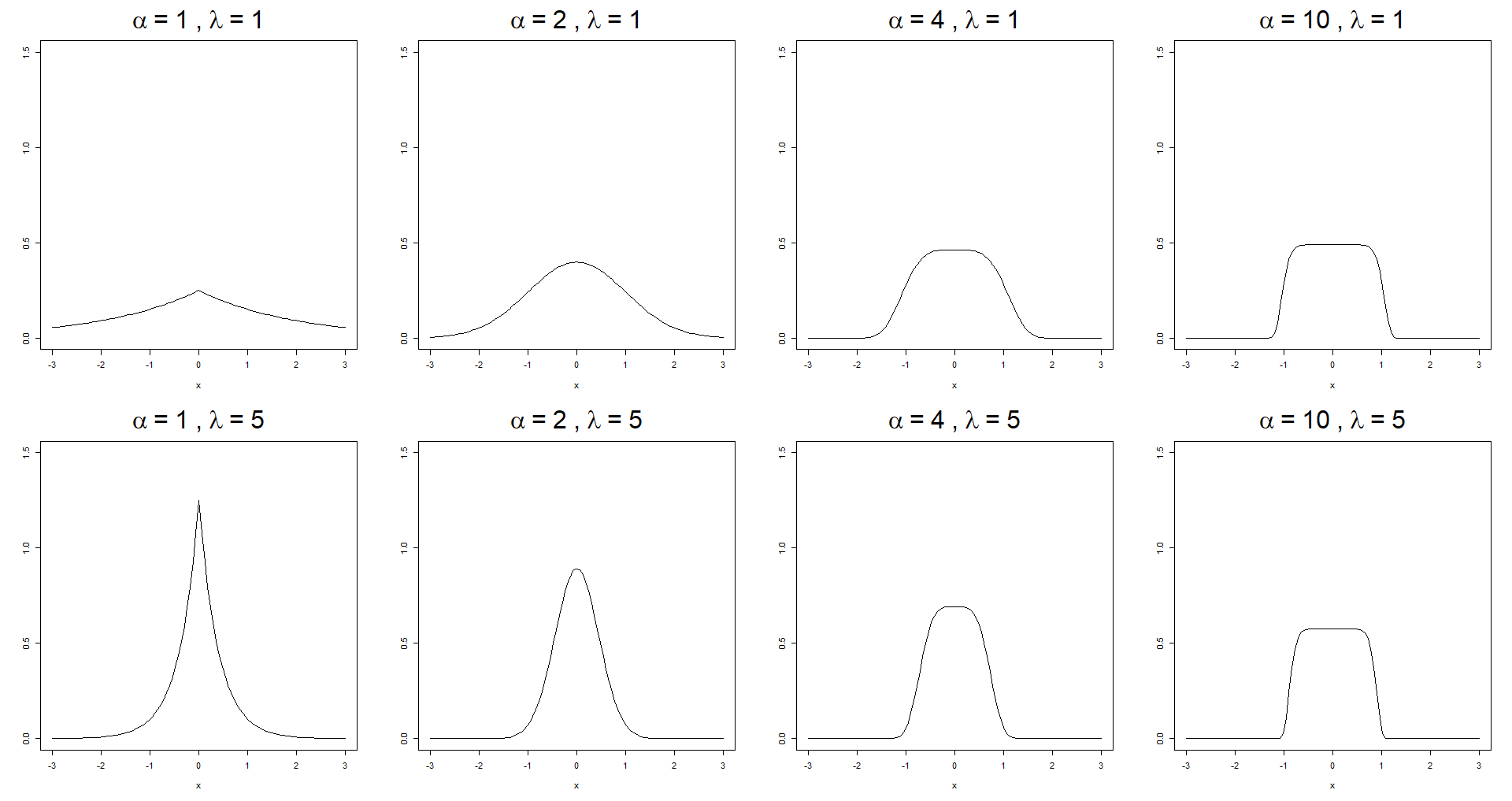}
    \caption{Exponential power prior on $\beta_j$ for different values of $\lambda$ and $\alpha$ when $\sigma^2=1$.}
    \label{fig:samplebetapriors}
\end{figure}

Figure \ref{fig:samplebetapriors} also highlights the role of $\lambda$ in the prior's behavior in sparse settings, where small $\alpha$ is preferred. The value of $\lambda$ must be large enough to shrink noise sufficiently but also small enough to avoid overshrinking large signals. Using a single value of $\lambda$ for all $\beta_j$ leads to sub-optimal performance; \cite{carvalhoetal2010} discuss this trade-off for the case $\alpha=1$. The proposed prior addresses these limitations by using a prior for $(\lambda,\alpha) \in (0,\infty)^p \times [k_1,k_2]$, where $0 < k_1 \le 1$ and $2 \le k_2$.  Allowing $k_1$ to be less than 1 will encourage sparsity when appropriate, while allowing $k_2$ to be larger than 2 will yield improved performance in dense settings. Replacing the scalar $\lambda$ with a $p$-dimensional vector allows for differing shrinkage in estimating each $\beta_j$.

As in the standard Bayesian formulation of penalized regression, the proposed prior assumes 
$$Y|X,\beta,\sigma^2 \sim N(X\beta,\sigma^2I_n).$$
A proper conjugate prior $\sigma^2 \sim \text{IG}(a,b)$ is assumed, and the prior on $\beta$ is
$$\nu(\beta|\lambda,\alpha,\sigma^2) = \left(\frac{\alpha}{2^{1/\alpha+1}(\sigma^2)^{1/\alpha}\Gamma(1/\alpha)}\right)^p\left(\prod_{j=1}^p\lambda_j\right)^{1/\alpha}\exp\left\lbrace-\frac{1}{2\sigma^2}\sum_{j=1}^p\lambda_j|\beta_j|^\alpha\right\rbrace.$$
The only difference from equation \eqref{eq:standardbayesbridge} is that each $\beta_j$ is assigned a parameter $\lambda_j > 0$. 
Observe that this corresponds to an exponential power prior on $\beta_j$ with $\mathbb{E}(\beta_j|\lambda_j,\sigma^2,\alpha) = 0$ and 
$$\var(\beta_j|\lambda_j,\sigma^2,\alpha)=\frac{\Gamma(3/\alpha)}{\Gamma(1/\alpha)}\left(\frac{\lambda_j}{\sigma^2}\right)^{-2/\alpha}4^{1/\alpha}.$$
Hence the variance is a decreasing function of $\lambda_j$. If $\lambda_j$ is small, larger values of $\beta_j$ are likely but if $\lambda_j$ is large, smaller values of $\beta_j$ are likely. Thus the prior on $\lambda_j$ should place a reasonable amount of mass close to zero while also allowing larger values. Since $\lambda_j$ is positive, this can be achieved with the following mixture of Gamma distributions,  
$$\lambda_j \sim \frac{1}{2}\text{Gamma}(e_{1},f_{1})+\frac{1}{2}\text{Gamma}(e_{2},f_{2}),$$
where the hyperparameters are chosen so that one component of the mixture concentrates its mass near zero, while the other is flatter with mass concentrated away from zero to accommodate large values. A simple approach to select such a prior is to choose $e_1$ and $f_1$ such that the first component has a small mean and variance and $e_2$ and $f_2$ such that the second component has a relatively large mean and variance.

Finally, a prior for $\alpha$ needs to be specified. Notice that unlike $\lambda_j$, which controls shrinkage for an individual $\beta_j$, this parameter is common to all of the $\beta_j$. If one wants to maintain the analogy with the frequentist methods in equation \eqref{eq:frequentistbridge}, the prior for $\alpha$ can be specified as a mixture of three components, where each represents the analyst's assessment of the relative importance of lasso, bridge, and ridge. Empirical work for such a mixture prior indicated that different choices for the mixture parameters yield similar estimation and prediction \citep{xiang2020}. This motivated consideration of a uniform distribution for $\alpha$ which has been found to work well, especially since extending the range of $\alpha$ appears to be impactful. Therefore it is assumed that
$$\alpha \sim \text{Unif}(k_1,k_2), \quad 0 < k_1 \le 1,\ 2 \le k_2.$$
The prior obtained for $\beta$ by marginalizing over $\lambda$, $\alpha$, and $\sigma^2$ is referred to as the generalized bridge (GBR) prior. 

\subsection{Effects of Hyperparameters}

Consider the effect of the hyperparameters on the GBR prior's density for a single regression coefficient when $\sigma^2=1$. Figure \ref{fig:gbr_k2} shows that larger values of $k_2$ increase both the mass assigned to the neighborhood around zero and the size of that neighborhood.
The value of $k_1$ has a minimal effect on these aspects of the density and instead determines the value at zero, with smaller values resulting in a taller spike (Table \ref{table:gbr_zero}). Thus larger values of $k_2$ are more suitable for dense settings, while small values of $k_1$ encourage sparsity.

\begin{table}[!ht]
    \centering
    \caption{Value of GBR density at zero for different values of $k_1$ when $k_2=2$ and $(e_1,f_1,e_2,f_2) = (0.1, 1, 2, 0.01)$.}
    \begin{tabular}{cc}
    \hline\rule{0pt}{2.5ex}  
    $k_1$ & $\nu(0)$ \\ \hline\rule{0pt}{3ex}
    0.01 & $5.105 \times 10^{166}$ \\ 
    0.1 & $7.068 \times 10^{14}$ \\ 
    0.5 & 151.337 \\ 
    1 & 7.851 \\ \hline
    \end{tabular}
    \label{table:gbr_zero}
\end{table}

\begin{figure}
    \centering
    \includegraphics[width=0.5\textwidth]{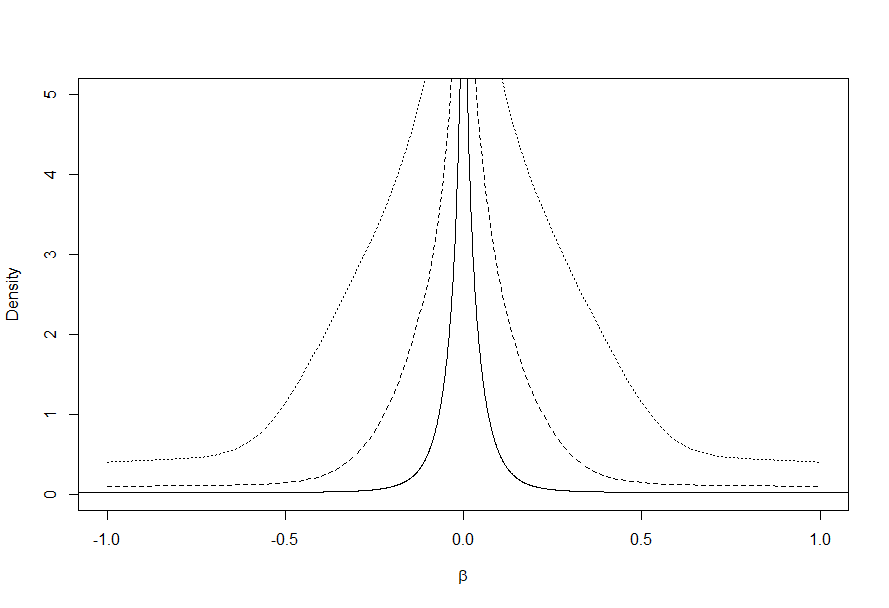}
    \caption{GBR density for $k_2=2$ (solid), $k_2=4$ (dashed), and $k_2=8$ (dotted) when $k_1=1$ and $(e_1,f_1,e_2,f_2) = (0.1, 1, 2, 0.01)$.}
    \label{fig:gbr_k2}
\end{figure}

Compared to $k_1$ and $k_2$, the values of $e_1$, $f_1$, $e_2$, and $f_2$ have more moderate effects on the density. Recall that these hyperparameters are specified so that the mean and variance of the first component of the prior on $\lambda_j$ are small and the mean and variance of the second component are large. Larger values of $e_1$ and smaller values of $f_1$ result in more mass near zero. Smaller values of $f_2$ result in more mass concentrated at zero. The effects of $e_2$ are most prominent, with values of $e_2\le 1$ resulting in flatter densities. Figure \ref{fig:gbr_densevssparse} compares the GBR density with $(e_1,f_1,e_2,f_2) = (0.1, 1, 2, 0.01)$ to the density with  $(e_1,f_1,e_2,f_2) = (1, 1, 40, 0.5)$ when $(k_1,k_2)=(0.5,4)$. The first prior has heavier tails and more mass concentrated at zero, while the second prior has more positive density on smaller, non-zero values. 
Figure \ref{fig:gbr_densevssparse_usual} compares the same two choices for $(e_1,f_1,e_2,f_2)$ with the usual range of $\alpha$, where $(k_1,k_2)=(1,2)$. The priors in Figure \ref{fig:gbr_densevssparse} have more mass concentrated at zero and small non-zero values and heavier tails, allowing better accommodation of sparse and dense settings. 

\begin{figure}
    \centering
    \begin{subfigure}[b]{0.45\textwidth}
         \centering
         \includegraphics[width=\textwidth]{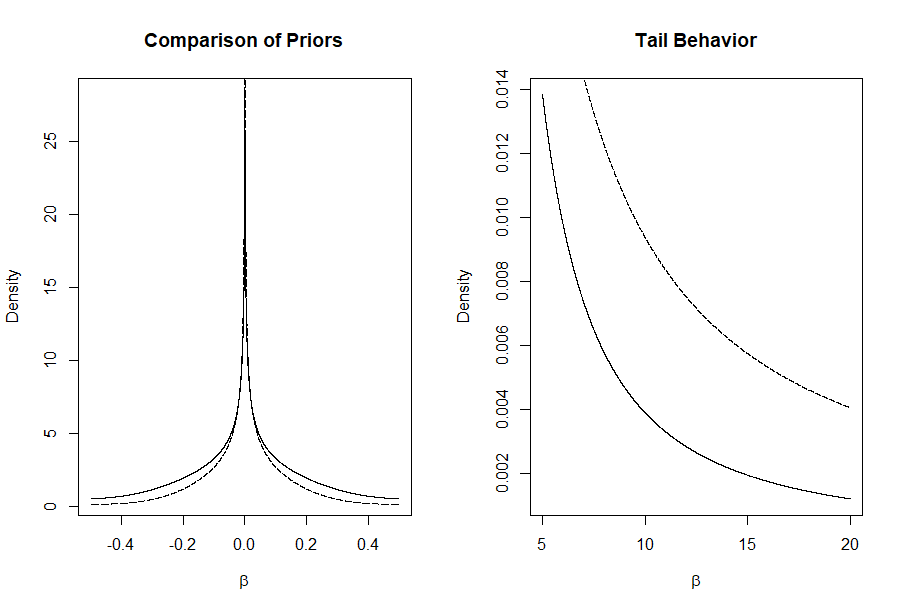}
         \caption{With $(k_1, k_2) = (0.5, 4)$.}
         \label{fig:gbr_densevssparse}
     \end{subfigure}
     \hfill
    \begin{subfigure}[b]{0.45\textwidth}
         \centering
         \includegraphics[width=\textwidth]{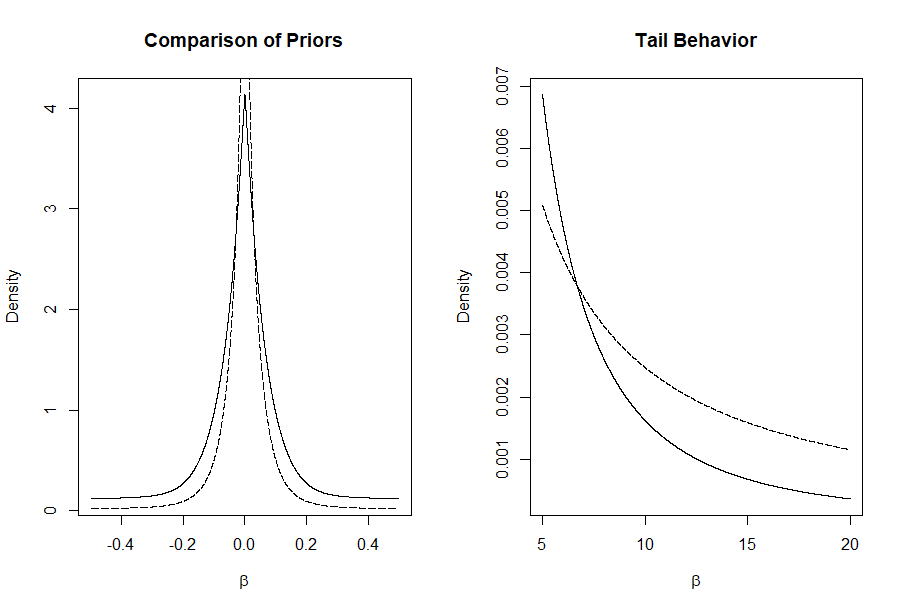}
         \caption{With the usual range $(k_1, k_2) = (1, 2)$.}
         \label{fig:gbr_densevssparse_usual}
    \end{subfigure}
    \hfill
    \caption{Comparison of the GBR densities with $(e_1,f_1,e_2,f_2) = (0.1, 1, 2, 0.01)$ (dashed) and $(e_1,f_1,e_2,f_2) = (1, 1, 40, 0.5)$ (solid). Code for reproducing these plots is available at the GitHub repository for the \texttt{R} package \texttt{GMCB}.}
\end{figure}

\subsection{Tail-Robustness}

The global-local structure of the GBR prior leads to the same sort of tail-robustness properties enjoyed by the horseshoe prior \citep{carvalhoetal2010}. Consider the following one-dimensional case of the model:
\begin{align*}
    Y|\beta & \sim N(\beta,1), \\
    \nu(\beta|\lambda,\alpha) &= \left(\frac{\alpha\lambda^{1/\alpha}}{2^{1/\alpha + 1}\Gamma(1/\alpha)}\right)\exp\left\lbrace-\frac{\lambda}{2}|\beta|^\alpha\right\rbrace, \\
    \lambda &\sim \frac{1}{2}\text{Gamma}(e_{1},f_{1})+\frac{1}{2}\text{Gamma}(e_{2},f_{2}), \\
    \alpha &\sim \text{Unif}(k_1,k_2).
\end{align*}
Let $m(y)$ be the marginal density achieved by integrating over all the parameters. A standard calculation shows that the marginal posterior mean of $\beta$ satisfies 
$$\mathbb{E}(\beta|y) = y + \frac{d}{dy}\log m(y),$$
and hence the following result shows that the GBR prior satisfies a tail-robustness property, indicating that bias is small for large signals.

\begin{theorem}
There is some $C_h$ which depends on the hyperparameters such that $|y-\mathbb{E}(\beta|y)| \le C_h$ and
$$\lim_{|y|\rightarrow\infty}\frac{d}{dy}\log m(y) = 0.$$
\end{theorem}
\begin{proof}
See Appendix \ref{app:tailrobustnessproof}.
\end{proof}

In estimation and prediction, the GBR prior is competitive with popular methods such as the horseshoe \citep{carvalhoetal2010} and the spike-and-slab lasso \citep{rockovageorge2018}. Simulation results and additional details for the univariate-response regression case, including sampling approaches and posterior consistency results, are available \citep{xiang2020}.

\section{The Generalized Mean-Covariance Bridge Prior}\label{sec:gmcb}
Consider now an extension of the GBR prior to mean and covariance estimation, where the response is modeled by the multivariate linear model in equation \eqref{eq:multimeancov}. The regression interpretation of the modified Cholesky decomposition is first reviewed \citep{pourahmadi1999} and used to re-express the Gaussian likelihood as a sequence of autoregressions. The following notation will be used. Let $j\mathbin{:}k$ denote the indices $j$ through $k$, and let $A^j$ denote column $j$ of matrix $A$. 

Define the matrix $T$ as in equation \eqref{eq:modcholprec}, and suppressing the dependence on $i$, equation \eqref{eq:multimeancov} can be re-expressed as $\varepsilon = T(Y - B^\top X)$. Because $T$ is lower triangular, the $j$th row equality is given by
$$
Y_j = (B^j)^\top X - T_{j,1} (Y_1 -(B^1)^\top X) - \cdots - T_{j, j-1} (Y_{j-1}-(B^{j-1})^\top X) + \varepsilon_j,
\quad j\in\{2,\ldots, q\},
$$
which is the linear regression of $Y_j$ on $Y_1,\ldots, Y_{j-1}$. Thus estimation of $(B,\Omega)$ can be accomplished by estimating a sequence of univariate-response regressions, where the $j$th regression has coefficients $B^j$ and $-T_{j,1:(j-1)}=\delta_j^\top$. 
Assuming $\varepsilon$ is multivariate normal with mean zero and diagonal covariance matrix $D$ with $D_{jj}=\gamma_j$, the likelihood of $Y$ can be expressed as 
\begin{align}
Y_1|X,B^1,\gamma_1 &\sim N\left((B^1)^\top X,\gamma_1\right), \nonumber \\
Y_j|Y_{1:(j-1)},X,B^{1:j},\delta_j,\gamma_j &\sim N\left((B^j)^\top X+\sum_{k=1}^{j-1}\delta_{j,k}\left(Y_k-(B^k)^\top X \right),\gamma_j\right),\quad j=2,\ldots,q. \label{eq:modelasreg}
\end{align}
The GBR prior in Section \ref{sec:generalizedbridge} can be directly specified on $B^j$ and $\delta_j$. 
Note that there is little reason to expect the same penalty to be appropriate for $B^j$ and $\delta_j$, as the matrices $B$ and $T$ are not expected to have similar levels of sparsity. Thus the GBR prior is specified such that elements of $B$ share a penalty parameter while elements of $T$ share a separate penalty parameter.
In particular, the prior on $B^j$ is
\begin{equation*}
    \nu(B^j|\Lambda^j,\alpha_b,\gamma_j) \propto \exp\left\lbrace-\frac{1}{2\gamma_j}\sum_{k=1}^p\lambda_{kj}|B_{kj}|^{\alpha_b}\right\rbrace,\quad j = 1,\ldots, q,
\end{equation*}
and the prior on $\delta_j$ is
\begin{equation*}
    \nu(\delta_j|\tau_j,\alpha_d,\gamma_j) \propto \exp\left\lbrace-\frac{1}{2\gamma_j}\sum_{k=1}^{j-1}\tau_{j,k}|\delta_{j,k}|^{\alpha_d}\right\rbrace, \quad j=2,\ldots,q.
\end{equation*}
For the remaining parameters, each of the regularization parameters $\lambda_{kj}$ and $\tau_{j,k}$ follow two-component Gamma mixture priors, the penalty parameters $\alpha_b$ and $\alpha_d$ are independent and identically distributed (i.i.d.) Unif$(k_1,k_2)$, and the $\gamma_j$ are i.i.d. IG$(a,b)$. The full hierarchical model is detailed in Appendix \ref{app:samplingalgorithms}, with the resulting posterior density characterized by equation \eqref{eq:posterior}.

The above prior is specified on $(B,T,D)$. The prior induced on $(B,\Omega)$ by this specification, marginalizing over the regularization parameters and the penalty parameters, is referred to as the generalized mean-covariance bridge (GMCB) prior. 
Similar to the GBR prior, when the regularization parameters are fixed and $\alpha_b$ and $\alpha_d$ are fixed to be equal, the posterior mode is the frequentist bridge estimate for $B^j$ and $\delta^j$.

\section{Estimation and Sampling}\label{sec:est}
For posterior inference on $(B,\Omega)$ under the GMCB prior, equation \eqref{eq:modcholprec} is used to obtain posterior samples of $\Omega$ from samples of $T$ and $D$. Point estimators are commonly selected to be the Bayes estimators under separate loss functions for $B$ and $\Omega$. A common choice of loss function is squared Frobenius loss for both $B$ and $\Omega$, resulting in the point estimator
\begin{align*}
    (\hat{B}_F,\hat{\Omega}_F) &=\Big(\mathbb{E}(B|Y,X),\ \mathbb{E}(\Omega|Y,X)\Big). 
\end{align*}
The following loss functions are also considered: 
\begin{align*}
    L_Q(\tilde{B},B)&=\tr\Big((\tilde{B}-B)\Omega(\tilde{B}-B)^\top\Big),\\
    L_S(\tilde{\Omega}, \Omega)&=\tr(\tilde{\Omega}\Omega^{-1})-\log|\tilde{\Omega}\Omega^{-1}|-q.
\end{align*}
The loss function $L_Q$ is the scalar quadratic loss \citep{yuasakubokawa2021} and is based on the Kullback-Liebler (KL) divergence between two matrix normal distributions with the same precision matrix, and $L_S$ is based on the KL divergence between two multivariate normal distributions with the same mean. 
The Bayes estimators under these loss functions \citep{yangberger1994,yuasakubokawa2021} are
\begin{align*}
    (\hat{B}_Q,\hat{\Omega}_S) &=\bigg(\mathbb{E}(B\Omega|Y,X)\Big[\mathbb{E}(\Omega|Y,X)\Big]^{-1},\  \Big[\mathbb{E}(\Omega^{-1}|Y,X)\Big]^{-1}\bigg). 
\end{align*}

Closed-form expressions under the GMCB prior are not available for either of these point estimators, requiring Monte Carlo methods for estimation.

\subsection{Markov Chain Monte Carlo}\label{sec:mcmc}
Two component-wise MCMC algorithms with invariant density characterized by equation \eqref{eq:posterior} in Appendix \ref{app:samplingalgorithms} are developed. The full details of the posterior conditionals are provided in Appendix \ref{app:postcond_mh}. Among these, only the  distributions for the $\lambda_{kj}$, $\tau_{j,k}$, and $\gamma_j$ are standard distributions. The GMCB-MH algorithm is a component-wise sampling scheme that uses Gibbs updates for these parameters and random walk Metropolis-Hastings updates with Gaussian proposal distributions for the remaining parameters. All updates are univariate, as block updates for $B^j$ and $\delta_j$ result in low acceptances rates, even for $p$ and $q$ small relative to $n$. For ease of notation, define $\gamma$ to be the vector of $\gamma_j$'s, and $\delta$ and $\tau$ to be vectors of length $q(q-1)/2$ concatenating the $\delta_j$ and $\tau_j$, respectively. The algorithm makes updates in the following order:
\begin{align*}
&\qquad(\Lambda,B,\alpha_b,\tau,\delta,\gamma,\alpha_d)\rightarrow\ (\Lambda',B,\alpha_b,\tau,\delta,\gamma,\alpha_d)\rightarrow\ (\Lambda',B',\alpha_b,\tau,\delta,\gamma,\alpha_d)\\
&\rightarrow\ (\Lambda',B',\alpha_b',\tau,\delta,\gamma,\alpha_d)\rightarrow\ (\Lambda',B',\alpha_b',\tau',\delta,\gamma,\alpha_d)\rightarrow\ (\Lambda',B',\alpha_b',\tau',\delta',\gamma,\alpha_d)\\
&\rightarrow\ (\Lambda',B',\alpha_b',\tau',\delta',\gamma',\alpha_d)\rightarrow\ (\Lambda',B',\alpha_b',\tau',\delta',\gamma',\alpha_d').
\end{align*}
This update order implies that initialization values for the regularization parameters $\Lambda$ and $\tau$ are not required. 

GMCB-MH is valid for all $n$, $p$, $q$, and $k_2\ge 2$. However, a more computationally efficient algorithm is available when $k_2=2$. In this case, the exponential power distribution has a SMN representation \citep{west1987}. This property can be leveraged to replace the Metropolis-Hastings updates for $B$ and $\delta$ in GMCB-MH with Gibbs updates. Although the SMN representation has been used for a data augmentation algorithm for the Bayesian bridge in multiple regression \citep{polsonetal2014}, the available implementation cannot be directly applied here because $B$ and $\delta$ are not conditionally independent.  

\subsection{GMCB-SMN Algorithm}

The GMCB-SMN algorithm updates $B$ and $\delta$ by augmenting their posterior conditional distributions. The basic strategy for $B$ is discussed here, with the full details and application of the strategy to sampling $\delta$ deferred to Appendix \ref{app:postcond_smn}. Both GMCB-SMN and GMCB-MH are implemented using the C++ interface provided in the \texttt{Rcpp} \citep{rcpp} and \texttt{RcppArmadillo} \citep{rcpparmadillo} packages in the \texttt{R} package \texttt{GMCB}, which is freely available at \url{https://github.com/czhao15103/GMCB}.

Let $p_{a}$ denote the density of a positive stable random variable with characteristic exponent $a < 1$. \cite{west1987} showed that the mixing density in the SMN representation of the exponential power distribution is the density of a polynomially-tilted positive stable random variable.  Then the SMN representation of the prior on $B_{kj}$ is given by 
\begin{align*}
    B_{kj}|\omega_{kj},\lambda_{kj},\gamma_j,\alpha_b & \sim N\left(0,\frac{1}{\omega_{kj}}\bigg(\frac{2\gamma_j}{\lambda_{kj}}\bigg)^{2/\alpha_b}\right),\\
    g(\omega_{kj}|\alpha_b)&\propto\omega_{kj}^{-1/2}\ p_{\alpha_b/2}(\omega_{kj}), \quad \omega_{kj} > 0,
\end{align*}
and the posterior conditional distribution of $B$ can be obtained by marginalizing over the $\omega_{kj}$:
\begin{align}
    q(B &|Y,X,\Lambda,\alpha_b,\delta,\gamma,\tau,\alpha_d) \nonumber \\
    &\propto \int_{\mathbb{R}_+^{pq}} \exp\left\lbrace-\frac{1}{2}\tr((Y-XB)\Omega(Y-XB)^\top)\right\rbrace \nonumber \\
    &\qquad\quad \cdot \left[\prod_{k=1}^p\prod_{j=1}^q \nu(B_{kj}|\omega_{kj},\lambda_{kj},\gamma_j,\alpha_b) g(\omega_{kj}|\alpha_b) \right]\,d\omega.\label{eq:baugment}
\end{align}
The integrand of equation \eqref{eq:baugment} will be referred to as the augmented posterior conditional distribution of $B$. To update $B$, the algorithm samples the conditional distributions associated with equation \eqref{eq:baugment} for $\omega$ and then $B$ and discards $\omega$. Doing so will preserve the invariant density of the Markov chain and does not affect the other posterior conditional distributions.

The conditional distribution of $\omega_{kj}$ associated with equation \eqref{eq:baugment} is the distribution of an exponentially-tilted positive stable random variable. \cite{devroye2009} proposed a double-rejection algorithm for sampling from this distribution that has been implemented in the \texttt{R} package \texttt{copula} \citep{copula}. A modified version of this implementation is used in the package \texttt{GMCB}.

For deriving the conditional distribution of $B$ associated with equation \eqref{eq:baugment}, note that the contribution from the SMN representation of the exponential power prior is easily expressed as a multivariate normal distribution on $\vect(B)$ with mean zero and a $pq \times pq$ diagonal covariance matrix $\Delta$. Thus the derivation is straightforward when the likelihood can be expressed in terms of $\vect(B)$. When $X^\top X$ is invertible, this can be accomplished by using the equivalence between a matrix normal distribution and a multivariate normal distribution. Routine calculation then shows that the conditional distribution under equation \eqref{eq:baugment} is a multivariate normal. The expressions for the mean and covariance matrix can be found in Appendix \ref{app:postcond_smn}.

When $p \ge n$, the likelihood cannot be rewritten using the equivalence between a matrix normal distribution and a multivariate normal distribution, as it relies on the invertibility of $X^\top X$. However, a variable transformation for $B$ using the SVD of $X$ and the modified Cholesky decomposition of $\Omega$ makes it possible to sample with a multivariate normal distribution. Define $U\in\mathbb{R}^{n\times n}$ and $V\in\mathbb{R}^{p\times p}$ to be orthonormal matrices and $C\in\mathbb{R}^{n\times p}$ such that the SVD of $X$ is given by $X=UCV^\top$. 
Define $\eta=V^\top B T^\top$. The Jacobian of this transformation is a constant with respect to $\eta$, and letting $\otimes$ denote the Kronecker product, the trace in the likelihood component can be rewritten in terms of $\eta$ as 
\begin{align*}
    -2\vect(C^\top U^\top YT^\top D^{-1})^\top \vect(\eta) + \vect(\eta)^\top (D^{-1/2}\otimes C)^\top (D^{-1/2}\otimes C)\vect(\eta).
\end{align*}
Although the matrix $(D^{-1/2}\otimes C)^\top (D^{-1/2}\otimes C)$ is not full rank, the trace in the prior after transformation is
$$ -\frac{1}{2}\vect(\eta)^\top (T^{-1} \otimes V)^\top  \Delta (T^{-1} \otimes V)\vect(\eta).$$
As $T$ is unit lower triangular and $V$ is orthonormal, $(T^{-1} \otimes V)^\top  \Delta (T^{-1} \otimes V)$ is positive definite, so that the conditional distribution of $\vect(\eta)$ is multivariate normal.

In either case, updating $B$ requires sampling from a $pq$-variate normal distribution. For sampling this distribution, the computational complexity of the approach in \cite{rue2001} is $O(p^3q^3)$; that is, there exists a constant $M>0$ such that the number of floating point operations required is bounded above by $Mp^3q^3$. Specific to GMCB-SMN, the computational complexity of the approach in \cite{bhattacharyaetal2016} is $O(np^2q^3)$, as the matrix $(T^{-1} \otimes V)^\top \Delta (T^{-1} \otimes V)$ is not sparse. To sample as efficiently as possible, GMCB-SMN uses the approach in \cite{rue2001} when $p < n$ and the approach in \cite{bhattacharyaetal2016} when $p\ge n$. As shown in Section \ref{sec:compeffort}, when $p \ll n$, this choice makes GMCB-SMN much faster than the HS-GHS algorithm, which always uses the approach in \cite{bhattacharyaetal2016} when updating $B$.

A similar augmentation strategy is used for the posterior conditional distribution of $\delta_j$. The SMN representation expresses the prior on $\delta_{j,k}$ as
\begin{align*}
    \delta_{j,k}|\epsilon_{j,k},\tau_{j,k},\gamma_j,\alpha_d&\sim N\left(0,\frac{1}{\epsilon_{j,k}}\bigg(\frac{2\gamma_j}{\tau_{j,k}}\bigg)^{2/\alpha_d}\right),\\
    g(\epsilon_{j,k}|\alpha_d)&\propto\epsilon_{j,k}^{-1/2}\ p_{\alpha_d/2}(\epsilon_{j,k}), \quad \epsilon_{j,k} > 0.
\end{align*}
Like the $\omega_{kj}$, the conditional distribution of each $\epsilon_{j,k}$ associated with the augmented posterior conditional distribution is an exponentially-tilted positive stable random variable. For the conditional distribution of $\delta_j$ associated with the augmented posterior conditional, standard calculations show that it is a $(j-1)$-variate normal distribution for all values of $n$, $p$, and $q$. When updating $\delta_j$, GMCB-SMN uses the approach in \cite{rue2001} when $j \le n$ and the approach in \cite{bhattacharyaetal2016} when $j > n$.

\subsubsection{Summarizing the GMCB-SMN Algorithm}

By augmenting the posterior conditional distributions of $B$ and $\delta$, GMCB-SMN only requires Metropolis-Hastings updates for the penalty parameters $\alpha_b$ and $\alpha_d$. Let $\omega$ and $\epsilon$ denote the collections of the latent scale variables for $B$ and $\delta$, respectively. The algorithm makes updates in the following order:
\begin{align*}
&\qquad(B,\Lambda,\alpha_b,\delta,\tau,\gamma,\alpha_d)\rightarrow\ (\omega',B,\Lambda,\alpha_b,\delta,\tau,\gamma,\alpha_d) \\
&\rightarrow\ (\omega',B',\Lambda,\alpha_b,\delta,\tau,\gamma,\alpha_d) \rightarrow\ (\omega',B',\Lambda',\alpha_b,\delta,\tau,\gamma,\alpha_d)\\
&\rightarrow\ (\omega',B',\Lambda',\alpha_b',\delta,\tau,\gamma,\alpha_d) \rightarrow\ (\omega',B',\Lambda',\alpha_b',\epsilon',\delta,\tau,\gamma,\alpha_d) \\
& \rightarrow\ (\omega',B',\Lambda',\alpha_b',\epsilon',\delta',\tau,\gamma,\alpha_d) \rightarrow(\omega',B',\Lambda',\alpha_b',\epsilon',\delta',\tau',\gamma,\alpha_d) \\
&\rightarrow\ (\omega',B',\Lambda',\alpha_b',\epsilon',\delta',\tau',\gamma',\alpha_d) \rightarrow\ 
(\omega',B',\Lambda',\alpha_b',\epsilon',\delta',\tau',\gamma',\alpha_d').
\end{align*}
The latent scale variables $\omega'$ and $\epsilon'$ are discarded, and the posterior conditional distributions used to update $\Lambda$, $\alpha_b$, $\tau$, $\gamma$, and $\alpha_d$ are the same as those used in GMCB-MH.

\subsection{Comparing GMCB-MH and GMCB-SMN}

Although GMCB-SMN is not applicable for values $k_2 > 2$, it is much more efficient than GMCB-MH. For the low-dimensional scenarios described in Section \ref{sec:sim}, sampler efficiency was compared based on the computation time and the multivariate effective sample size (ESS) \citep{vatsetal2019}, which was computed using the \texttt{R} package \texttt{mcmcse} \citep{mcmcse}, for 1e5 iterations. As shown in Table \ref{table:esscomparison} in Section \ref{sec:sim}, GMCB-MH has a much lower multivariate ESS. This is likely due to higher autocorrelation for $B$ and $\delta$, as illustrated in Figure \ref{fig:acfcomparison}. Autocorrelation for the remaining parameters tended to be similar between the two algorithms, which is expected as the sampling approach is the same for those parameters. While the difference in computation time shown in Figure \ref{fig:gmcb_time_comparison} is not that large, changes in the relative size of $p$ and $q$ can have a significant impact, which will be discussed further in Sections \ref{sec:sim} and \ref{sec:wisp}.

\begin{figure}
    \centering
    \begin{subfigure}[b]{0.45\textwidth}
         \centering
         \includegraphics[width=\textwidth]{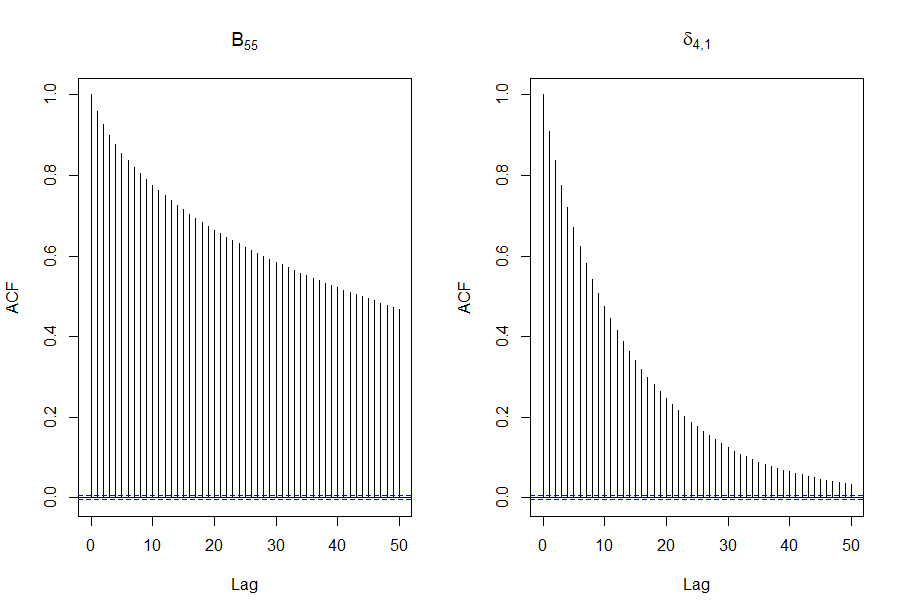}
         \caption{GMCB-MH.}
     \end{subfigure}
     \hfill
    \begin{subfigure}[b]{0.45\textwidth}
         \centering
         \includegraphics[width=\textwidth]{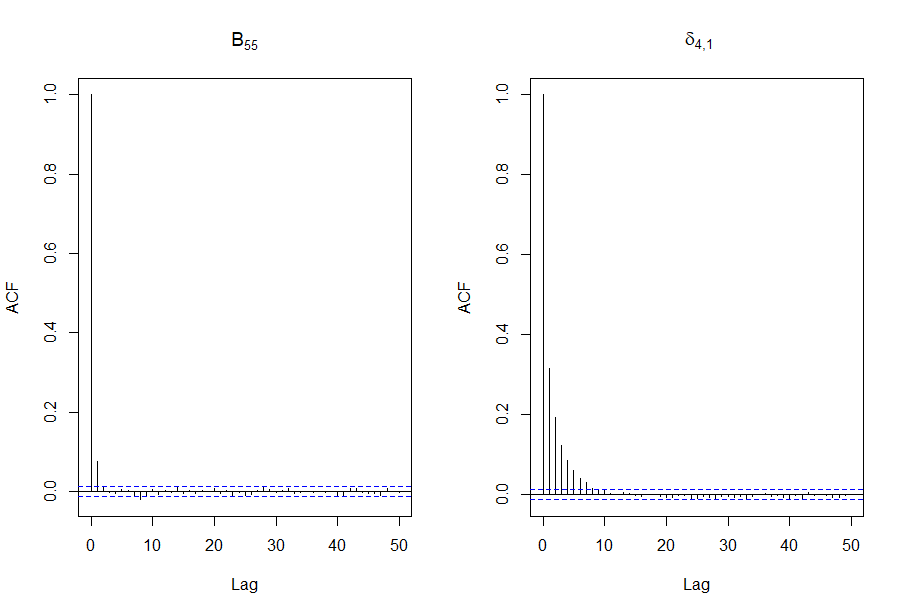}
         \caption{GMCB-SMN.}
     \end{subfigure}
    \caption{Autocorrelation plots for randomly selected elements of $B$ and $\delta$ in Scenario 1 in Section \ref{sec:sim}. These plots are typical of what was observed in each scenario. Additional correlation plots for the other scenarios are available in the appendix.}
    \label{fig:acfcomparison}
\end{figure}

\begin{figure}
    \centering
    \includegraphics[width=0.7\textwidth]{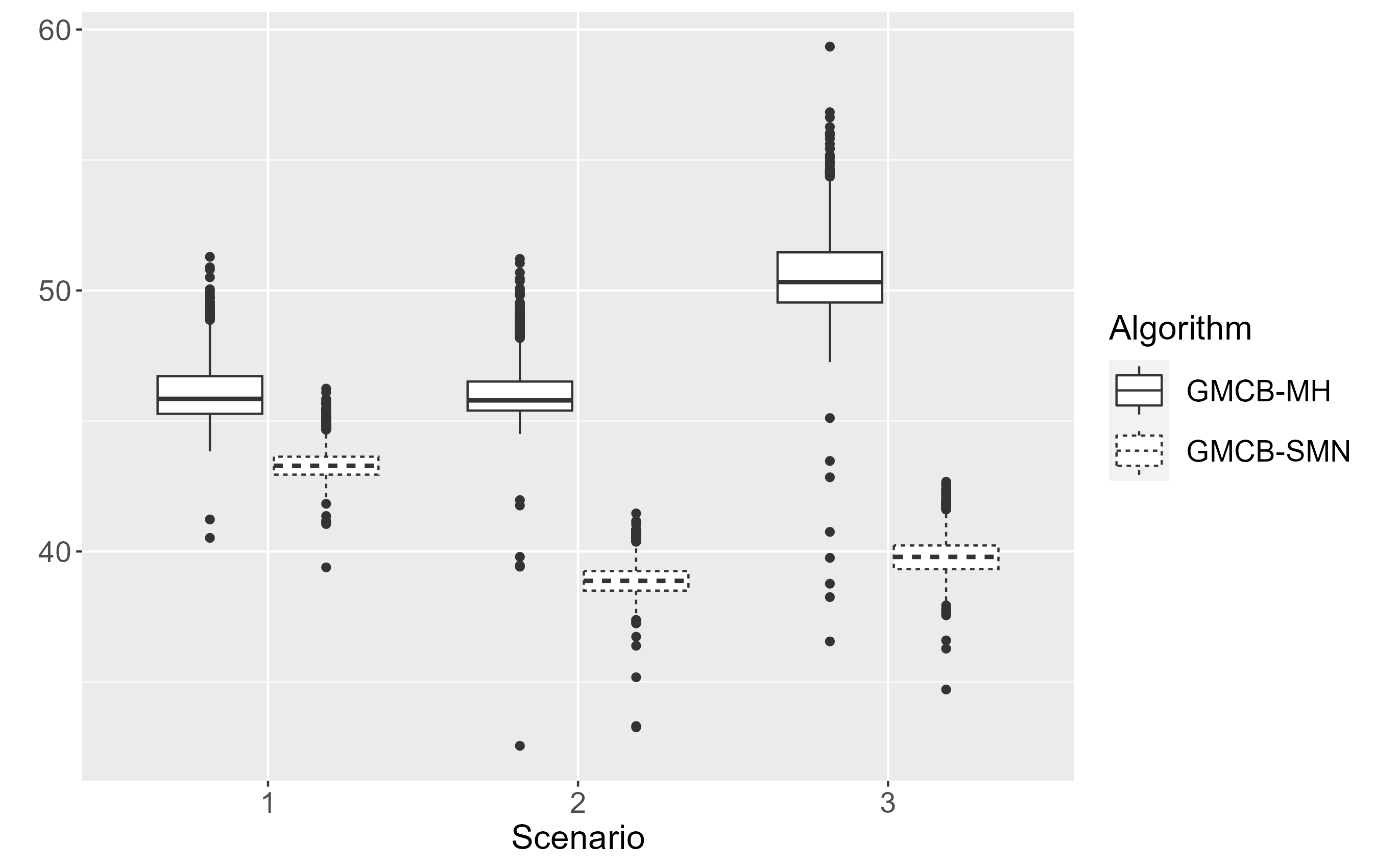}
    \caption{Comparison of GMCB-MH and GMCB-SMN based on total computation time in seconds for $100,000$ iterations, based on 2000 replications.}
    \label{fig:gmcb_time_comparison}
\end{figure}

\begin{remark}
    While it is possible to reduce the dimension of the posterior by integrating out $\Lambda$ and $\tau$, doing so does not result in any computational advantages -- in fact, the elements of $\gamma$ can no longer be sampled using standard distributions.
\end{remark}

\begin{remark}
A data augmentation algorithm based on a scale mixture of uniforms (SMU) representation is also possible. The SMU representation has been used for the Bayesian bridge in multiple regression by \cite{griffinhoff2020} and \cite{mallickyi2018}. Under this representation, truncated normal distributions are used to update $B$ and $\delta$. Although the algorithm is valid for $k_2\ge 2$ and allows block Gibbs sampling of $B$, Gibbs sampling of $\delta$ under the SMU representation requires $q\le n$. This restriction can be addressed using a SVD-based transformation, but it would require $q-n$ SVDs to be computed at every iteration of the algorithm when $q > n$.
\end{remark}

\section{Simulation Experiments}\label{sec:sim}
Estimates for $(B,\Omega)$ from GMCB are compared with estimates from HS-GHS \citep{lietal2021} and mSSL \citep{deshpandeetal2019}. Frequentist coverage of the 95\% posterior credible intervals produced by GMCB is examined in Appendix \ref{app:postci}. The following simulations compare the maximum a posteriori (MAP) estimates from both optimization algorithms for mSSL (DPE and DCPE) with the Bayes estimates $(\hat{B}_F,\hat{\Omega}_F)$ and $(\hat{B}_Q,\hat{\Omega}_S)$ from GMCB and HS-GHS. 
For ease of reference, Table \ref{table:abbreviations} lists the algorithms and their corresponding abbreviations.
The MATLAB code by \cite{lietal2021} was used for HS-GHS, and the \texttt{R} package \texttt{mSSL} \citep{mSSL} was used for mSSL. Code for replicating the simulations can be found at the GitHub repository for the \texttt{R} package \texttt{GMCB}.

\begin{table}[!ht]
    \centering
    \begin{tabular}{p{0.7\linewidth}p{0.2\linewidth}} \hline 
     Algorithm & Abbreviation \\ \hline
      Standard Metropolis-within-Gibbs algorithm for the generalized mean-covariance bridge prior & GMCB-MH \\
      Data augmentation algorithm for the generalized mean-covariance bridge prior & GMCB-SMN \\
      Multivariate spike-and-slab LASSO with dynamic posterior exploration & mSSL-DPE \\
      Multivariate spike-and-slab LASSO with dynamic conditional posterior exploration & mSSL-DCPE \\
      Horseshoe-graphical horseshoe & HSGHS \\
      Graphical horseshoe & GHS \\
      Graphical spike-and-slab LASSO & gSSL \\
      Maximum likelihood estimator & MLE \\ \hline
    \end{tabular}
    \caption{Algorithm names and their abbreviations.}
    \label{table:abbreviations}
\end{table}

\subsection{Low-dimensional Scenarios}\label{sec:lowdimsim}

In the following scenarios, $n=100$, $p=5$, and $q=5$. For each scenario, the rows of $X$ were drawn independently from $N_p(0,\Sigma_X)$, where the $ij$th element of $\Sigma_X$ is $0.7^{|i-j|}$, and 2000 responses were generated from the model
$$Y=XB+E, \quad E\sim MVN_{n,p}(0,I_n,\Omega^{-1}=\Sigma).$$
Before estimation, the design matrix was standardized and the response matrix was centered. For hyperparameter settings, the default values in the package \texttt{mSSL} were used for mSSL-DPE and mSSL-DCPE. For both GMCB algorithms, $k_1=0.5$, while $k_2=2$ for GMCB-SMN and $k_2=4$ for GMCB-MH. An empirical prior specified using the method-of-moments was used for $\gamma$. The priors on the regularization parameters differed between scenarios. Extensive empirical work indicated that a Gamma$(1, 1)$ and Gamma$(40, 0.5)$ mixture prior seems to be a reasonable starting point. The results presented use this prior for both $\Lambda$ and $\tau$ in Scenario 1 and for $\tau$ in Scenario 3. A Gamma$(0.1, 1)$ and Gamma$(2, 0.01)$ mixture prior was used for both $\Lambda$ and $\tau$ in Scenario 2 and for $\Lambda$ in Scenario 3.

The multivariate ESS approach \citep{vatsetal2019} was used to determine the number of iterations for the GMCB-MH, GMCB-SMN, and HS-GHS samplers. The minimum ESS for the 40 parameters of interest as computed using the \texttt{mcmcse} package \citep{mcmcse} was 8438, and the number of iterations was selected so that the multivariate ESS was approximately twice the minimum ESS.

\paragraph{Scenario 1:} The entries of the matrix $B$ were independently drawn from $N(2,0.001^2)$. The covariance matrix is compound symmetric with $\Sigma_{ij}=0.7^{I\{i\ne j\}}$, so $\Omega$ is a dense matrix. The parameters of the modified Cholesky decomposition are
\begin{align*}
    \gamma &= (1, 0.51, 0.424, 0.388, 0.368)^\top
\end{align*}
and 
\begin{align*}
    \delta_{j,k} &= \begin{cases}
    \begin{aligned}
        0.7,\quad &j = 2,\\
        0.412,\quad &j = 3,\\
        0.292,\quad &j = 4, \\
        0.226,\quad &j = 5.
    \end{aligned}
    \end{cases}
\end{align*}

GMCB-MH, GMCB-SMN, and the HS-GHS algorithm were run for $1.5\times 10^5$, $2.5\times 10^4$, and $2.5\times 10^4$ iterations, respectively.

\paragraph{Scenario 2:} The entries of the matrix $B$ were independently drawn from $N(5,1^2)$, and 12 entries were randomly set to zero. The covariance matrix has an AR(1) structure with $\Sigma_{ij}=0.7^{|i-j|}$, so $\Omega$ is banded. The parameters of the modified Cholesky decomposition are
\begin{align*}
    \gamma &= (1, 0.51, 0.51, 0.51, 0.51)^\top
\end{align*}
and 
\begin{align*}
    \delta_{j,k} &= \begin{cases}
    \begin{aligned}
        0.7,\quad &k = j-1,\\
        0,\quad &\text{otherwise.}
    \end{aligned}
    \end{cases}
\end{align*}

GMCB-MH, GMCB-SMN, and the HS-GHS algorithm were run for $1.5\times 10^5$, $2.5\times 10^4$, and $2.5\times 10^4$ iterations, respectively.

\paragraph{Scenario 3:} Three randomly selected coefficients in $B$ were independently drawn from $N(15,3^2)$, and the remaining coefficients were set to zero. The precision matrix was defined as in \cite{danielspourahmadi2002} scenario IIIA:
\begin{align*}
    \gamma &= (0.5, 0.7, 1, 3, 5)^\top
\end{align*}
and
\begin{align*}
    \delta_{j,k} &= \begin{cases}
    \begin{aligned}
        0.75+0.02k,\quad &k = j-1,\\
        0.4,\quad &k = j-2,\\
        0.2,\quad &k = j-3, \\
        0.1,\quad &k = j-4.
    \end{aligned}
    \end{cases}
\end{align*}
This results in a nonstationary covariance matrix.

GMCB-MH, GMCB-SMN, and the HS-GHS algorithm were run for $1.5\times 10^5$, $2.75\times 10^4$, and $3\times 10^4$ iterations, respectively. 

\vspace{4mm}

The average Frobenius loss for Scenarios 1--3 is displayed in Table \ref{table:simresults}. (The results under $L_S$ and $L_Q$, which show similar trends, are omitted). 
Estimation accuracy for $B$ is similar across all scenarios for the three models.
For estimation of $\Omega$, the average loss for mSSL was nearly double that of the fully Bayesian approaches when $\Omega$ was dense (Scenarios 1 and 3). 
When $\Omega$ was sparse (Scenario 2), mSSL slightly outperformed the fully Bayesian approaches.
The lack of local regularization parameters in mSSL may explain the difference in performance. mSSL uses the regularization parameter from the slab component of the spike-and-slab LASSO prior on the off-diagonal elements as the rate of the exponential prior on the diagonal elements of $\Omega$. In both Scenarios 1 and 3, the diagonal elements of $\Omega$ are much larger in magnitude than the off-diagonal elements, while the magnitudes are similar in Scenario 2. Without local parameters, the same amount of regularization is applied to both the diagonal and non-zero off-diagonal elements.

Comparing the results for GMCB-MH and GMCB-SMN, the utility of allowing larger values for $k_2$ in dense scenarios is evident in the estimation of $\Omega$. In Scenarios 1 and 3, the $\delta_j$ are dense with small signals, and extending the upper bound to $k_2=4$ places more prior mass on smaller, non-zero values, allowing GMCB-MH to outperform both GMCB-SMN and HS-GHS in both scenarios, while GMCB-SMN only outperforms HS-GHS in Scenario 3. When the $\delta_j$ are sparse (Scenario 2), extending the upper bound has little benefit, but both GMCB methods outperform HS-GHS. Increasing $k_2$ has little effect on estimation of $B$, even in Scenario 1, where $B$ is dense.

The maximum likelihood estimates (MLE) are also included in Table \ref{table:simresults} as a benchmark. None of the Bayesian methods outperform the MLE in estimation of $B$. However, the MLE for $\Omega$ is always significantly outperformed by GMCB and HS-GHS. As observed when comparing mSSL to the fully Bayesian approaches, mSSL outperforms the MLE only in Scenario 2, when $\Omega$ is sparse.

\begin{table}[!ht]
    \centering
    \caption{Average squared Frobenius loss of $\tilde{B}$ and $\tilde{\Omega}$ for GMCB, HS-GHS, and mSSL, based on 2000 replications. The maximum standard error for all values in the table was 0.073.}
    \begin{tabular}{ccccccccc}
    \hline\rule{0pt}{2.5ex}  
    Scenario & $(\tilde{B},\tilde{\Omega})$ & Method & $\lVert\Tilde{B}-B\rVert_F^2$ & $\lVert\Tilde{\Omega}-\Omega\rVert_F^2$ \\ \hline\rule{0pt}{3ex}
    \multirow{9}{*}{1} & \multirow{3}{*}{$(\hat{B}_F,\hat{\Omega}_F)$} & GMCB-MH & 0.944 & 2.454 \\
    & & GMCB-SMN & 0.944 & 3.612 \\
    & & HS-GHS & 0.920 & 3.471 \\ \cline{2-5}\rule{0pt}{3ex}
    & \multirow{3}{*}{$(\hat{B}_Q,\hat{\Omega}_S)$} & GMCB-MH & 0.941 & 2.287 \\
    & & GMCB-SMN & 0.940 & 3.247 \\
    & & HS-GHS & 0.918 & 2.787 \\ \cline{2-5}\rule{0pt}{3ex}
    & \multirow{2}{*}{MAP} & mSSL-DCPE & 0.898 & 5.925 \\
    & & mSSL-DPE & 0.899 & 5.972 \\ \cline{2-5}\rule{0pt}{3ex}
    & & MLE & 0.709 & 4.024 \\ \hline\rule{0pt}{3ex}
    \multirow{9}{*}{2} & \multirow{3}{*}{$(\hat{B}_F,\hat{\Omega}_F)$} & GMCB-MH & 1.286  & 1.706 \\
    & & GMCB-SMN & 1.227 & 1.265 \\ 
    & & HS-GHS & 1.437 & 1.862 \\ \cline{2-5}\rule{0pt}{3ex}
    & \multirow{3}{*}{$(\hat{B}_Q,\hat{\Omega}_S)$} & GMCB-MH & 1.292 & 1.439 \\
    & & GMCB-SMN & 1.229 & 1.081 \\ 
    & & HS-GHS & 1.445 & 1.589 \\ \cline{2-5}\rule{0pt}{3ex}
    & \multirow{2}{*}{MAP} & mSSL-DCPE & 1.195 & 1.038 \\ 
    & & mSSL-DPE & 1.194 & 1.045 \\ \cline{2-5}\rule{0pt}{3ex}
    & & MLE & 0.585 & 3.757 \\ \hline\rule{0pt}{3ex}
    \multirow{9}{*}{3} & \multirow{3}{*}{$(\hat{B}_F,\hat{\Omega}_F)$} & GMCB-MH & 3.331 & 0.741 \\
    & & GMCB-SMN & 3.293 & 0.827 \\ 
    & & HS-GHS & 3.338 & 0.934 \\ \cline{2-5}\rule{0pt}{3ex}
    & \multirow{3}{*}{$(\hat{B}_Q,\hat{\Omega}_S)$} & GMCB-MH & 3.339 & 0.678 \\
    & & GMCB-SMN & 3.298 & 0.747 \\ 
    & & HS-GHS & 3.346 & 0.761 \\ \cline{2-5}\rule{0pt}{3ex}
    & \multirow{2}{*}{MAP} & mSSL-DCPE & 3.413 & 1.614 \\ 
    & & mSSL-DPE & 3.402 & 1.524 \\ \cline{2-5}\rule{0pt}{3ex}
    & & MLE & 2.616 & 1.221 \\ \hline
    \end{tabular}
    \label{table:simresults}
\end{table}

\subsection{High-dimensional Scenarios}

The following scenarios assess the performance of GMCB in high-dimensional settings. Scenario 4 considers a multivariate linear regression setting, while Scenario 5 considers a mean-covariance estimation with no covariates.  
Due to the computational burden required to estimate multivariate ESS accurately in high-dimensional settings, the number of iterations for the GMCB-SMN and HS-GHS algorithms was fixed at 50,000 in Scenarios 4 and 5. The GMCB-MH algorithm is omitted due to the difficulty of tuning the Metropolis-Hastings steps in high dimensions.
For the mSSL algorithms, the maximum number of iterations permitted in Scenario 4 was 100,000 instead of the default 500. As the available implementation of mSSL is designed for multivariate linear regression and cannot perform an intercept-only multivariate linear regression, it was excluded from Scenario 5 and replaced with an estimator for $(B,\Omega)$ using the sample mean and the graphical spike-and-slab lasso (gSSL), the mean-zero counterpart of mSSL. An estimator based on the sample mean and the precision matrix estimate from the graphical horseshoe \citep{lietal2019}, or GHS, is also included.

\paragraph{Scenario 4:} The dimensions of the problem are $n=40$, $p=30$, and $q=50$. The covariate matrix was generated as in Scenarios 1--3, and 100 response matrices were generated from the multivariate linear regression model. 

Five percent of the elements of $B$ were randomly selected to be non-zero. The magnitudes of these elements were independently drawn from a Unif$(0.5,2)$ distribution, and the signs were independently drawn from a Ber$(0.5)$ distribution. 

The precision matrix was defined as in the clique setting of \cite{lietal2021}. Sixteen cliques with three members were randomly generated. The variance of $Y_i$ was set to 1, and for $i\ne j$, $\Omega_{ij}=0.75$ for $i$ and $j$ in the same clique and $\Omega_{ij}=0$ otherwise. Approximately 4\% of the strictly lower-triangular elements of $T$ are non-zero, with 17 out of the 49 $\delta_j$ vectors being zero vectors, 16 having one non-zero element, and the remaining 16 having two non-zero elements.
Note that this structure does not correspond to a precision matrix commonly associated with a naturally ordered response.

The priors used for GMCB-SMN were a mixture of Gamma$(0.1,1)$ and Gamma$(2,0.01)$ for $\Lambda$ and $\tau$ and Unif $(0.5,2)$ for $\alpha_b$ and $\alpha_d$.

\paragraph{Scenario 5:} The dimensions of the problem are $n=40$, $p=1$, and $q=50$. One hundred datasets were generated. For each dataset, $n$ i.i.d. observations were drawn from $N_q(B^\top,\Omega^{-1})$. The mean vector $B^\top$ was generated as follows. A vector was first drawn from a $q$-dimensional, mean-zero normal distribution with a compound symmetric covariance matrix with $ij$th element equal to $0.5^{I\{i\ne j\}}$. The sorted vector was then set equal to $B^\top$.

The covariance matrix $\Omega^{-1}$ was generated as in example 3 of \cite{leelee2021}, with 
\begin{equation*}
    \Omega^{-1}_{ij} = \frac{1}{2}\left( \lvert \lvert i - j \rvert + 1 \rvert^{1.4} - 2\lvert i - j \rvert^{1.4} + \lvert\lvert i -j \rvert - 1 \rvert^{1.4} \right)
\end{equation*}
This corresponds to a long-range dependence structure where $\Omega$ is not sparse but is a diagonally-dominant matrix.

The priors used for GMCB-SMN were a mixture of Gamma$(1,1)$ and Gamma$(40,0.5)$ for the prior on $\Lambda$, a mixture of Gamma$(0.1,1)$ and Gamma$(2,0.01)$ for the prior on $\tau$, and  Unif$(0.5,2)$ for $\alpha_b$ and $\alpha_d$.
\vspace{4mm}

Table \ref{table:highdimmulti} displays the average Frobenius loss for Scenario 4.
This scenario is expected to be challenging for GMCB, because the assumption of a naturally ordered response does not hold. However, estimation of $B$ by $\hat{B}_F$ was comparable for GMCB-SMN and HS-GHS, and GMCB-SMN outperforms both mSSL algorithms in estimating $B$. Furthermore, despite the lack of a natural ordering, GMCB-SMN significantly outperforms HS-GHS and mSSL in estimation of $\Omega$. 

The average Frobenius loss for Scenario 5 is displayed in Table \ref{table:highdimmeancov}. Although the natural ordering assumption made by GMCB is satisfied in this scenario, GMCB is outperformed by HS-GHS. Unlike the previous scenarios, this scenario does not include covariates. The difference in relative performance between this scenario and the previous scenarios suggests that GMCB greatly benefits from the presence of covariates when estimating $\Omega$. In addition, the difference in estimation error for $\Omega$ lies in estimation of the off-diagonal elements excluding the superdiagonal and subdiagonal. 
These elements range in value from $-0.074$ to $-0.002$, so strong shrinking of these elements will not greatly increase Frobenius loss when $\Omega$ is estimated directly as in HS-GHS, gSSL, and GHS. However, when $\Omega$ is estimated indirectly through a decomposition as in GMCB, estimation error of the factors may be magnified when estimating $\Omega$. 

\begin{table}[!ht]
    \centering
    \caption{Average squared Frobenius loss of $\tilde{B}$ and $\tilde{\Omega}$ for GMCB, HS-GHS, and mSSL, based on 100 replications for Scenario 4. The maximum standard error for all values in the table was 0.138.}
    \begin{tabular}{ccccccccc}
    \hline\rule{0pt}{2.5ex}  
    Scenario & $(\tilde{B},\tilde{\Omega})$ & Method & $\lVert\Tilde{B}-B\rVert_F^2$ & $\lVert\Tilde{\Omega}-\Omega\rVert_F^2$ \\ \hline\rule{0pt}{3ex}
    \multirow{6}{*}{4} & \multirow{2}{*}{$(\hat{B}_F,\hat{\Omega}_F)$} & 
     GMCB-SMN & 26.10 & 29.46 \\
    & & HS-GHS & 25.77 & 37.21 \\ \cline{2-5}\rule{0pt}{3ex}
    & \multirow{2}{*}{$(\hat{B}_Q,\hat{\Omega}_S)$} & 
    GMCB-SMN & 38.47 & 28.14 \\
    & & HS-GHS & 26.89 & 45.58 \\ \cline{2-5}\rule{0pt}{3ex}
    & \multirow{2}{*}{MAP} & mSSL-DCPE & 43.14 & 69.72 \\
    & & mSSL-DPE & 45.52 & 70.42 \\ \hline
    \end{tabular}
    \label{table:highdimmulti}
\end{table}

\begin{table}[!ht]
    \centering
    \caption{Average squared Frobenius loss of $\tilde{B}$ and $\tilde{\Omega}$ for GMCB, HS-GHS, and mSSL, based on 100 replications for Scenario 5. The maximum standard error for all values in the table was 0.051.}
    \begin{tabular}{ccccccccc}
    \hline\rule{0pt}{2.5ex}  
    Scenario & $(\tilde{B},\tilde{\Omega})$ & Method & $\lVert\Tilde{B}-B\rVert_F^2$ & $\lVert\Tilde{\Omega}-\Omega\rVert_F^2$ \\ \hline\rule{0pt}{3ex}
    \multirow{6}{*}{5} & \multirow{2}{*}{$(\hat{B}_F,\hat{\Omega}_F)$} &
    GMCB-SMN & 2.504 & 23.368 \\
    & & HS-GHS & 1.565 & 15.302 \\ \cline{2-5}\rule{0pt}{3ex}
    & \multirow{2}{*}{$(\hat{B}_Q,\hat{\Omega}_S)$} & 
    GMCB-SMN & 1.619 & 12.930 \\
    & & HS-GHS & 1.380 & 11.770 \\ \cline{2-5}\rule{0pt}{3ex}
    & \multirow{3}{*}{Samp. Mean + $\tilde{\Omega}$} & gSSL-MAP & 1.275 & 13.002 \\ 
    & & GHS-$\hat{\Omega}_F$  & 1.275 & 20.590 \\ 
    & & GHS-$\hat{\Omega}_S$ & 1.275 & 12.062 \\ \hline
    \end{tabular}
    \label{table:highdimmeancov}
\end{table}

\subsection{Comparison of Computational Effort}\label{sec:compeffort}
The computational efficiency of the GMCB and HS-GHS algorithms were compared based on the computation time and multivariate ESS for 1e5 iterations. The mSSL algorithms, which produce MAP estimates rather than posterior samples, are omitted from this comparison.
For the scenarios in Section \ref{sec:lowdimsim}, the multivariate ESS of HS-GHS is comparable to that of GMCB-SMN in Scenario 1 but lower in the other two scenarios (Table \ref{table:esscomparison}). 
The total computation time of HS-GHS is much higher in all three scenarios (Table \ref{table:compcomparison}), because it always uses the approach of \cite{bhattacharyaetal2016} for sampling multivariate normals when updating $B$.

However, when $p > n$, the linear scaling of computational complexity with $p$ for HS-GHS makes it much faster than GMCB-SMN, which has computational complexity $O(np^2q^3)$ when using the method of \cite{bhattacharyaetal2016}. Table \ref{table:compcomparison_highdim} compares the computation time for 500 iterations for $n=100$, $p = 120$, and $q = 50$. Because 500 iterations are insufficient for effectively estimating over 7000 parameters, estimation results for this scenario are omitted. As expected, the required computation time is significantly higher for GMCB-SMN. Note that GMCB-MH is faster than both GMCB-SMN and HS-GHS for a fixed number of iterations, but achieving an acceptable multivariate ESS likely requires far more iterations for GMCB-MH.

\begin{table}[!ht]
    \centering
    \caption{Average multivariate ESS for GMCB-MH, GMCB-SMN, and HS-GHS for the scenarios in Section \ref{sec:sim}, based on 2000 replications. The maximum standard error for all values in the table was 1.916.}
    \begin{tabular}{ccccc}
    \hline\rule{0pt}{2.5ex}  
    Scenario & $(\tilde{B},\tilde{\Omega})$ & GMCB-MH & GMCB-SMN & HS-GHS \\ \hline\rule{0pt}{3ex}
    \multirow{2}{*}{1} & $(\hat{B}_F,\hat{\Omega}_F)$ & 13470 & 77180 & 77400 \\ 
    & $(\hat{B}_Q,\hat{\Omega}_S)$ & 13250 & 81720 & 82400  \\ \hline\rule{0pt}{3ex}
    
    \multirow{2}{*}{2} & $(\hat{B}_F,\hat{\Omega}_F)$ & 13572 & 83320 & 75670 \\ 
    & $(\hat{B}_Q,\hat{\Omega}_S)$ & 14070 & 83100 & 79300  \\ \hline\rule{0pt}{3ex}
    
    \multirow{2}{*}{3} & $(\hat{B}_F,\hat{\Omega}_F)$ & 13760 & 72040 & 59800 \\ 
    & $(\hat{B}_Q,\hat{\Omega}_S)$ & 14070 & 71480 & 60400  \\  \hline
    \end{tabular}
    \label{table:esscomparison}
\end{table}

\begin{table}[!ht]
    \centering
    \caption{Average computation time in seconds for GMCB-MH, GMCB-SMN, and HS-GHS for the scenarios in Section \ref{sec:sim}, based on 2000 replications. The maximum standard error for all values in the table was 0.0677.}
    \begin{tabular}{cccc}
    \hline\rule{0pt}{2.5ex}  
    Scenario & GMCB-MH & GMCB-SMN & HS-GHS \\ \hline\rule{0pt}{3ex}
    1$\ $ & 46.11 & 43.33 & 1062 \\ 
    2 & 46.08 & 38.90 & 1060 \\ 
    3 & 50.55 & 39.80 & 1284 \\ \hline
    \end{tabular}
    \label{table:compcomparison}
\end{table}

\begin{table}[!ht]
    \centering
    \caption{Average computation time in hours for GMCB-MH, GMCB-SMN, and HS-GHS for a scenario with $n=100$, $p = 120$, $q = 50$, based on 50 replications. The maximum standard error for all values in the table was 0.037.}
    \begin{tabular}{cccc}
    \hline\rule{0pt}{2.5ex}  
    GMCB-MH & GMCB-SMN & HS-GHS \\ \hline\rule{0pt}{3ex}
    2.6 & 46 & 6.3 \\ \hline
    \end{tabular}
    \label{table:compcomparison_highdim}
\end{table}
 
\section{The WISP Survey Example}\label{sec:wisp}
The following dataset provides an example of a naturally-ordered response and demonstrates the computational feasibility of the GMCB model for mean-covariance estimation on a large $n$ and large $q$ dataset without covariates. As such, the focus is not analysis. If this dataset were analyzed to answer a research question of interest, appropriate methods should be applied. GMCB may be considered when mean and covariance estimation is necessary for such methods.

The WFC3 Infrared Spectroscopic Parallel (WISP) Survey is a pure parallel Hubble Space Telescope (HST) program, meaning the data were collected by the Wide Field Camera 3 (WFC3) while other HST instruments were in use. The survey used the WFC3's two near-infrared grisms\footnote{Grisms are combinations of a diffraction grating and a prism that can produce dispersed images of light spectra for all objects in the field of view \citep{grisms}.} \citep{ateketal2010},
which cover different wavelengths -- the $G_{102}$ covers the 800-1150 nanometer (nm) range, while the $G_{141}$ covers the 1075-1700 nm range \citep{wfc3}. By combining the spectra from the two grisms, it is possible to detect multiple emission lines for each object \citep{ateketal2010}. The detection and identification of these lines are necessary for the scientific goals of the WISP survey, as they provide the distance of the galaxies and allow the measurement of galaxies' physical properties \citep{dickinsonetal2018}. Visual inspection of these spectra would be time-consuming: the WISP survey has observed 483 fields\footnote{Data is only available for 432 \citep{dickinsonetal2018}.} \citep{baronchellietal2020}, and each field includes hundreds, or even thousands, of spectra. 

One-dimensional spectra (i.e., the flux, or brightness, at each observed wavelength) are extracted from the dispersed images and used for the detection of emission lines. In these spectra, spurious features, e.g., due to contamination by stellar diffraction spikes, nearby continuum sources, or zeroth orders, may be mistaken for emission lines \citep{ateketal2010}. However, visual inspection of the two-dimensional images and extracted spectra provides classifications for the emission lines as either genuine or spurious \citep{dickinsonetal2018}.

The one-dimensional spectra are an example of the type of data GMCB is most suited for analyzing. The flux measurements are ordered by the observed wavelengths, satisfying the assumption of a naturally-ordered response. Although HS-GHS and mSSL do not make such an assumption, these methods can also be considered when mean and covariance estimation is necessary for answering questions of interest. 

To prepare the data for this example, the spectra considered were restricted to fields that were covered by both grisms. Data processing involved re-arranging the flux measurements so that they were ordered by wavelength and appropriately combining measurements for the same wavelength. After this processing, $135,386$ spectra from 258 fields, each with $q=253$ unique wavelengths, remained. These spectra were then divided into classes based on the presence or absence of emission lines as determined by a WISP emission-line catalog constructed by \cite{bagleyetal2020}. Based on the catalog, $129,529$ of the $135,386$ spectra were classified as having no emission lines and $3565$ were classified as having emission lines. The remaining $2292$ spectra were excluded from the analysis, as the degree of agreement between reviewers in identifying the emission lines was low. Due to class imbalance, each class was further divided into a training and test sample, with 70\% of the class in the training sample. This resulted in $90,670$ spectra with no emission lines and $2495$ spectra with emission lines in the training sample. For a method to be successfully applied to this dataset, it must be able to accommodate a dataset with $90,670$ observations and a response with dimension $q=253$. 

Covariates were not included in this example, and GMCB, HS-GHS, and mSSL were considered for application to the dataset. As previously mentioned, the available implementation of mSSL is unable to estimate a mean-only model. For HS-GHS, the algorithm could not be run with $n=90,670$, $p=1$, and $q=253$. Updating the mean would require inversion of a 22,939,510 by 22,939,510 ($nq \times nq$) matrix on each iteration of the algorithm, which could not be completed with 128 GB of RAM. For the GMCB algorithms, use of GMCB-MH is difficult, as it requires 32,133 component-wise Metropolis-Hastings updates, and the time required to tune the proposal step sizes and obtain estimates was prohibitive. Furthermore, the computation time required to run GMCB-MH was much longer than the time required for GMCB-SMN. GMCB-MH required an average of 6--8 hours for 1000 iterations, while GMCB-SMN required only 20--30 minutes for 2000 iterations.

\section{Final Remarks}\label{sec:finalremarks}
Simultaneous mean and covariance estimation under a multivariate linear model was proposed using a novel penalized regression prior. The GBR prior allows the practitioner to address uncertainty in the regularization and penalty parameters in a principled manner by averaging over the posterior distribution. This is useful for simultaneous mean and covariance estimation, where there may be little prior information about the appropriate amount of penalization for the covariance or precision matrix. In the development of GMCB-SMN, selection of the sampling strategy for $B$ based on the relative size of $p$ and $n$ enables the algorithm to handle data with large $n$, as demonstrated by the WISP survey example, and makes it much faster than the available HS-GHS algorithm for $p \ll n$.

There are a few extensions of this work that may be of particular interest for analyzing longitudinal data. It is common to measure covariates at each time point, so that the covariates associated with a single response vector $Y$ may be a $q\times p$ matrix rather than a single vector of covariates. The mean regression coefficient $B$ is then a $p$-dimensional vector instead of a $p\times q$ matrix. The Gaussian likelihood can still be reparameterized into a sequence of autoregressions in this case, though efficient sampling of $\delta$ and $B$ is likely to be more challenging than in the framework considered here.

Consider also the case where the covariance matrix of $Y$ is known to be a sum of two positive definite matrices $A$ and $E$. The regression framework used here can be extended to such cases by observing that if $Y\sim N_q(\mu, A + E)$, this distribution is the marginal distribution of $Y$ in the hierarchy
\begin{align*}
    Y|Z&\sim N_q(\mu, E), \\
    Z&\sim N_q(0,A).
\end{align*}
The joint distribution of $Y$ and $Z$ is the product of two normal distributions, each of which can be rewritten in the regression framework of the modified Cholesky decomposition. Thus both $A$ and $E$ can be estimated using a prior such as the GMCB prior. This hierarchy can be generalized to fit linear mixed models, as well as error-in-variable linear regression models.

\section*{Acknowledgement}

The authors thank Claudia Scarlata and Hugh Dickinson for providing the WISP survey data. The authors would also like to thank Claudia Scarlata for helpful discussions on the WISP survey example. Jones was partially supported by NSF grant DMS-2152746.

\FloatBarrier
\bibliographystyle{apalike}
\bibliography{References}

\appendixpage
\begin{appendix}
\section{Proof of Theorem 1}\label{app:tailrobustnessproof}
\begin{proof}
Notice that
\begin{align*}
    m(y) = &\frac{1}{k_2 - k_1}\int \int \frac{1}{\sqrt{2\pi}} \exp\left(-\frac{(\beta-y)^2}{2}\right) \frac{\alpha}{2^{1/\alpha+1}\Gamma(1/\alpha)} \\
    & \qquad\qquad\quad \cdot \left\lbrace \frac{1}{2} \frac{f_1^{e_1}}{\Gamma(e_1)}\frac{\Gamma(e_1 + 1/\alpha)}{(|\beta|^\alpha/2 + f_1)^{e_1 + 1/\alpha}} + \frac{1}{2} \frac{f_2^{e_2}}{\Gamma(e_2)}\frac{\Gamma(e_2 + 1/\alpha)}{(|\beta|^\alpha/2 + f_2)^{e_2 + 1/\alpha}} \right\rbrace \,d\beta \,d\alpha.
\end{align*}
For ease of computation assume that $e_1 = e_2 = e$, $f_1 = f_2 = f$ and therefore
\begin{align*}
    m(y) = & \frac{1}{k_2 - k_1}\int \int \frac{1}{\sqrt{2\pi}} \exp\left(-\frac{(\beta-y)^2}{2}\right) \frac{\alpha}{2^{1/\alpha+1}\Gamma(1/\alpha)} \\
    & \qquad\qquad\qquad \cdot \frac{f^{e}}{\Gamma(e)}\frac{\Gamma(e + 1/\alpha)}{(|\beta|^\alpha/2 + f)^{e + 1/\alpha}} \,d\beta \,d\alpha
\end{align*}
and
\begin{align*}
    m'(y) = & \frac{1}{k_2 - k_1}\int \int \frac{1}{\sqrt{2\pi}} (\beta - y) \exp\left(-\frac{(\beta-y)^2}{2}\right) \frac{\alpha}{2^{1/\alpha+1}\Gamma(1/\alpha)} \\
    & \qquad\qquad\qquad \cdot \frac{f^{e}}{\Gamma(e)}\frac{\Gamma(e + 1/\alpha)}{(|\beta|^\alpha/2 + f)^{e + 1/\alpha}} \,d\beta \,d\alpha.
\end{align*}
Set
\begin{align*}
    g(y) &= \int \exp\left(-\frac{(\beta-y)^2}{2}\right) \frac{1}{(|\beta|^\alpha/2 + f)^{e + 1/\alpha}} \,d\beta \\
    h(y) &= \int (\beta-y) \exp\left(-\frac{(\beta-y)^2}{2}\right) \frac{1}{(|\beta|^\alpha/2 + f)^{e + 1/\alpha}} \,d\beta.
\end{align*}
Let $t = \beta - y$ so that
\begin{align*}
    g(y) &= \int \exp\left(-\frac{t^2}{2}\right) \frac{1}{(|t + y|^\alpha/2 + f)^{e + 1/\alpha}} \,dt \\
    h(y) &= \int t \exp\left(-\frac{t^2}{2}\right) \frac{1}{(|t + y|^\alpha/2 + f)^{e + 1/\alpha}} \,dt
\end{align*}
while if $t = - t$ then 
\begin{align*}
    g(y) &= \int \exp\left(-\frac{t^2}{2}\right) \frac{1}{(|y - t|^\alpha/2 + f)^{e + 1/\alpha}} \,dt \\
    h(y) &= - \int t \exp\left(-\frac{t^2}{2}\right) \frac{1}{(|y - t|^\alpha/2 + f)^{e + 1/\alpha}} \,dt.
\end{align*}
Hence
\begin{align*}
    2g(y) &= \int \exp\left(-\frac{t^2}{2}\right) \Big[(|y + t|^\alpha/2 + f)^{-e - 1/\alpha} + (|t - y|^\alpha/2 + f)^{-e - 1/\alpha}\Big] \,dt \\
    2h(y) &= \int t \exp\left(-\frac{t^2}{2}\right) \Big[(|y + t|^\alpha/2 + f)^{-e - 1/\alpha} - (|t - y|^\alpha/2 + f)^{-e - 1/\alpha}\Big] \,dt.
\end{align*}
Notice that both functions under the integral sign are even. Thus
\begin{align*}
    g(y) &= \int_0^\infty \exp\left(-\frac{t^2}{2}\right) \Big[(|y + t|^\alpha/2 + f)^{-e - 1/\alpha} + (|t - y|^\alpha/2 + f)^{-e - 1/\alpha}\Big] \,dt \\
    h(y) &= \int_0^\infty t \exp\left(-\frac{t^2}{2}\right) \Big[(|y + t|^\alpha/2 + f)^{-e - 1/\alpha} - (|t - y|^\alpha/2 + f)^{-e - 1/\alpha}\Big] \,dt.
\end{align*}
Suppose $y > 0$ (a nearly identical proof will hold with $y < 0$). Then
\begin{align}
    g(y) & > \int_0^\infty \exp\left(-\frac{t^2}{2}\right) (|t - y|^\alpha/2 + f)^{-e-1/\alpha}\,dt \nonumber \\
    & > \int_0^{y/2} \exp\left(-\frac{t^2}{2}\right) (|t - y|^\alpha/2 + f)^{-e-1/\alpha}\,dt \nonumber \\
    & = \int_0^{y/2} \exp\left(-\frac{t^2}{2}\right) ((y - t)^\alpha/2 + f)^{-e-1/\alpha}\,dt \nonumber \\
    & > \int_0^{y/2} \exp\left(-\frac{t^2}{2}\right) (y^\alpha/2 + f)^{-e-1/\alpha}\,dt \nonumber \\
    & = C_1 \cdot (y^\alpha/2 + f)^{-e-1/\alpha} \label{eq:proofgylb},
\end{align}
where $C_1 < \frac{\sqrt{2\pi}}{2}$. Next consider
\begin{align*}
    -h(y) &= \int_0^\infty t \exp\left(-\frac{t^2}{2}\right) \Big[(|t - y|^\alpha/2 + f)^{-e - 1/\alpha} - (|y + t|^\alpha/2 + f)^{-e - 1/\alpha} \Big] \,dt \\
    &= h_1(y) + h_2(y),
\end{align*}
where
\begin{align*}
    h_1(y) &= \int_0^{y/2} t \exp\left(-\frac{t^2}{2}\right) \Big[((y - t)^\alpha/2 + f)^{-e - 1/\alpha} - ((y + t)^\alpha/2 + f)^{-e - 1/\alpha} \Big] \,dt \\
    h_2(y) &= \int_{y/2}^\infty t \exp\left(-\frac{t^2}{2}\right) \Big[(|t - y|^\alpha/2 + f)^{-e - 1/\alpha} - (|y + t|^\alpha/2 + f)^{-e - 1/\alpha} \Big] \,dt.
\end{align*}
Set
\begin{equation*}
    S(t) = ((y - t)^\alpha/2 + f)^{-e-1/\alpha} - ((y + t)^\alpha/2 + f)^{-e-1/\alpha}.
\end{equation*}
The next step is to show that, when $0 < t < y/2$,
\begin{equation*}
    S(t) < V(t) = 4\alpha(e + 1/\alpha)((y/2)^\alpha/2 + f)^{-e-1/\alpha}\frac{t}{y}.
\end{equation*}
First notice that $S(0) = V(0) = 0$ so that it suffices to show that $S'(t) < V'(t)$. Consider
\begin{align*}
    S'(t) &= \frac{\alpha(e + 1/\alpha)}{2}(y - t)^{\alpha - 1}\left(\frac{(y - t)^\alpha}{2} + f\right)^{-(e + 1/\alpha + 1)} \\
    &\qquad + \frac{\alpha(e + 1/\alpha)}{2}(y + t)^{\alpha - 1}\left(\frac{(y + t)^\alpha}{2} + f\right)^{-(e + 1/\alpha + 1)} \\
    V'(t) &= ((y/2)^\alpha/2 + f)^{-e-1/\alpha}\cdot\frac{1}{y}\cdot 4\alpha(e + 1/\alpha).
\end{align*}
Notice that
$$f(x) = x^{\alpha - 1}\left(\frac{x^\alpha}{2} + f\right)^{-(e + 1/\alpha + 1)} = x^{-(2 + \alpha e)}\left(\frac{1}{2} + \frac{f}{x^\alpha}\right)^{-(e + 1 + 1/\alpha)}$$
is a decreasing function when $x$ is large and $0 < k_1 < \alpha < k_2\le 4$. Thus, when $y$ is large and $0 < t < y/2$, $y - t > y/2$ and $y + t > y > y/2$. Therefore,
\begin{align*}
    yS'(t) &= \frac{\alpha(e + 1/\alpha)}{2}y(y - t)^{\alpha - 1}\left(\frac{(y - t)^\alpha}{2} + f\right)^{-(e + 1/\alpha + 1)} \\
    &\qquad + \frac{\alpha(e + 1/\alpha)}{2}y(y + t)^{\alpha - 1}\left(\frac{(y + t)^\alpha}{2} + f\right)^{-(e + 1/\alpha + 1)} \\
    & < \alpha(e + 1/\alpha)y(y/2)^{\alpha - 1}\left(\frac{(y/2)^\alpha}{2} + f\right)^{-(e + 1/\alpha + 1)} \\
    & < \alpha(e + 1/\alpha)y(y/2)^{\alpha - 1}\left(\frac{(y/2)^\alpha}{2}\right)^{-1}\left(\frac{(y/2)^\alpha}{2} + f\right)^{-(e + 1/\alpha)} \\
    & = 4\alpha(e + 1/\alpha)((y/2)^\alpha/2 + f)^{-e-1/\alpha},
\end{align*}
which gives $S'(t) < V'(t)$. So $S(t) < V(t)$. Then
\begin{align*}
    h_1(y) &= \int_0^{y/2} t\exp\left(-\frac{t^2}{2}\right)S(t)\,dt \\
    &< \int_0^{y/2} t^2\exp\left(-\frac{t^2}{2}\right)((y/2)^\alpha/2 + f)^{-e-1/\alpha}\cdot \frac{1}{y}\cdot 4\alpha(e + 1/\alpha)\,dt \\
    &<((y/2)^\alpha/2 + f)^{-e-1/\alpha}\cdot \frac{1}{y}\cdot 4\alpha(e + 1/\alpha) \int_{-\infty}^\infty t^2\exp\left(-\frac{t^2}{2}\right)\,dt \\
    &< ((y/2)^\alpha/2 + f)^{-e-1/\alpha}\cdot \frac{1}{y}\cdot 4\alpha(e + 1/\alpha)\cdot 2\sqrt{2\pi}
\end{align*}
and
\begin{align*}
    h_2(y) &= \int_{y/2}^\infty t\exp\left(-\frac{t^2}{2}\right) \Big[(| t- y|^\alpha/2 + f)^{-e-1/\alpha} - (|y + t|^\alpha/2 + f)^{-e-1/\alpha}\Big] \,dt\\
    &< \int_{y/2}^\infty t\exp\left(-\frac{t^2}{2}\right) (| t- y|^\alpha/2 + f)^{-e-1/\alpha} \,dt \\
    & < f^{-e-1/\alpha} \int_{y/2}^\infty t\exp\left(-\frac{t^2}{2}\right) \,dt \\
    &= f^{-e-1/\alpha} \exp\left(-\frac{y^2}{8}\right).
\end{align*}
Therefore,
\begin{align}
    -h(y) &= h_1(y) + h_2(y) \nonumber\\
    & < ((y/2)^\alpha/2 + f)^{-e-1/\alpha}\cdot \frac{1}{y}\cdot 4\alpha(e + 1/\alpha)\cdot 2\sqrt{2\pi} + f^{-e-1/\alpha} \exp\left(-\frac{y^2}{8}\right). \label{eq:proofhyupper}
\end{align}
By equations \eqref{eq:proofgylb} and \eqref{eq:proofhyupper},
\begin{align*}
    0 &< -\cfrac{m'(y)}{m(y)} = \cfrac{\int -h(y) \cfrac{\alpha \Gamma(e + 1/\alpha)}{2^{1/\alpha}\Gamma(1/\alpha)}\,d\alpha}{\int g(y) \cfrac{\alpha \Gamma(e + 1/\alpha)}{2^{1/\alpha}\Gamma(1/\alpha)}\,d\alpha} \\
    &< \frac{1}{y}\cfrac{\int \cfrac{\alpha \Gamma(e + 1/\alpha)}{2^{1/\alpha + 1}\Gamma(1/\alpha)} \Big[((y/2)^\alpha/2 + f)^{-e-1/\alpha}8\alpha(e+1/\alpha)\Big]\,d\alpha}{\int \cfrac{\alpha \Gamma(e + 1/\alpha)}{2^{1/\alpha + 1}\Gamma(1/\alpha)}C_1(y^\alpha/2 + f)^{-e-1/\alpha}(2\pi)^{-1/2}\,d\alpha} \\
    &\qquad + \cfrac{\int \cfrac{\alpha \Gamma(e + 1/\alpha)}{2^{1/\alpha + 1}\Gamma(1/\alpha)} \Big[ f^{-e-1/\alpha}(2\pi)^{-1/2}\exp\left(-\frac{y^2}{8}\right)\Big]\,d\alpha}{\int \cfrac{\alpha \Gamma(e + 1/\alpha)}{2^{1/\alpha + 1}\Gamma(1/\alpha)}C_1(y^\alpha/2 + f)^{-e-1/\alpha}(2\pi)^{-1/2}\,d\alpha} \\
    &\equiv \frac{1}{y}A_1 + A_2.
\end{align*}
Notice that $A_1$ is bounded by a constant since
\begin{align*}
    A_1 &= \cfrac{\int \cfrac{\alpha \Gamma(e + 1/\alpha)}{2^{1/\alpha + 1}\Gamma(1/\alpha)} \Big[((y/2)^\alpha/2 + f)^{-e-1/\alpha}8\alpha(e+1/\alpha)\Big]\,d\alpha}{\int \cfrac{\alpha \Gamma(e + 1/\alpha)}{2^{1/\alpha + 1}\Gamma(1/\alpha)}C_1(y^\alpha/2 + f)^{-e-1/\alpha}(2\pi)^{-1/2}\,d\alpha} \\
    &\le \cfrac{\int \cfrac{k_2 \Gamma(e + 1/k_1)}{2^{1/k_2 + 1}\Gamma(1/k_2)} \Big[((y/2)^\alpha/2 + f)^{-e-1/\alpha}8k_2(e+1/k_1)\Big]\,d\alpha}{\int \cfrac{k_1 \Gamma(e + 1/k_2)}{2^{1/k_1 + 1}\Gamma(1/k_1)}C_1(y^\alpha/2 + f)^{-e-1/\alpha}(2\pi)^{-1/2}\,d\alpha} \\
    &\equiv C_{e,f,k_1,k_2}^{(0)}\cfrac{\int((y/2)^\alpha/2 + f)^{-e-1/\alpha}\,d\alpha}{\int(y^\alpha/2 + f)^{-e-1/\alpha}\,d\alpha} \\
    &< C_{e,f,k_1,k_2}^{(0)} \cfrac{\int(y^\alpha/2 + f)^{-e-1/\alpha}\,d\alpha}{\int(y^\alpha/2 + f)^{-e-1/\alpha}\,d\alpha} = C_{e,f,k_1,k_2}^{(0)}.
\end{align*}
Notice that, because of the term $\exp\lbrace-y^2/8\rbrace$, $A_2$ is a higher order term of $A_1$ when $y$ goes to infinity. Then 
$$\frac{1}{y}A_1 + A_2 = \frac{1}{y}C^{(1)}_{e,fk_1,k_2} + o\left(\frac{1}{y}\right)C^{(2)}_{e,f,k_1,k_2}\sim O(1/y)$$
since $C^{(1)}_{e,fk_1,k_2} < C^{(0)}_{e,fk_1,k_2}$ and $C^{(2)}_{e,fk_1,k_2}$ are constants depending on the choice of $k_1, k_2,e,f$. Therefore $\lim_{y\rightarrow\infty}m'(y)/m(y)=0$.
\end{proof}

\section{Sampling Algorithms for GMCB}\label{app:samplingalgorithms}
Recall that $j\mathbin{:}k$ denotes the indices $j$ through $k$ and that $A^j$ denotes column $j$ of matrix $A$. In the following, IG$(a,b)$ and Gamma$(a,b)$ denote the shape-rate parameterization of the inverse Gamma and Gamma distributions, respectively.

The fully specified model under the GMCB prior is as follows. The likelihood is given by
\begin{equation*}
Y_i|X_i,B,\Omega^{-1}\sim N_q(B^\top X_i,\Omega^{-1}),\qquad i = 1,\ldots,n.
\end{equation*}
Under the modified Cholesky decomposition, this is equivalent to 
\begin{align*}
    Y^1|X,B^1,\gamma_1 &\sim N_n\left(XB^1,\gamma_1I_n\right), \\
    Y^j|Y^{1:(j-1)},X,B^{1:j},\delta_j,\gamma_j &\sim N_n\left(XB^j+\left(Y^{1:(j-1)}-XB^{1:(j-1)} \right)\delta_j,\gamma_jI_n\right),\quad j=2,\ldots,q. 
\end{align*}
The prior on $B^j$ for $j=1,\ldots,q$ is given by
\begin{align*}
\nu(B^j|\Lambda^j,\alpha_b,\gamma_j) &= \left(\frac{\alpha_b}{2^{1/\alpha_b+1}\gamma_j^{1/\alpha_b}\Gamma(1/\alpha_b)}\right)^{p}\left(\prod_{k=1}^p\lambda_{kj}\right)^{1/\alpha_b}\exp\left\lbrace-\frac{1}{2\gamma_j}\sum_{k=1}^p\lambda_{kj}|B_{kj}|^{\alpha_b}\right\rbrace,\\
\lambda_{kj}&\sim\frac{1}{2}\text{Gamma}(e_{kj,1},f_{kj,1})+\frac{1}{2}\text{Gamma}(e_{kj,2},f_{kj,2}), \\
\alpha_b &\sim\text{Unif}(k_1,k_2), \quad 0 < k_1\le1,\ k_2\ge 2,
\end{align*}
and the prior for $\delta_j$ for $j=2,\ldots,q$ is given by
\begin{align*}
\nu(\delta_j|\tau_j,\gamma_j,\alpha_d)&=\left(\frac{\alpha_d}{2^{1/\alpha_d+1}\gamma_j^{1/\alpha_d}\Gamma(1/\alpha_d)}\right)^{j-1}\left(\prod_{k=1}^{j-1}\tau_{j,k}\right)^{1/\alpha_d}\exp\left\lbrace-\frac{1}{2\gamma_j}\sum_{k=1}^{j-1}\tau_{j,k}|\delta_{j,k}|^{\alpha_d}\right\rbrace,\\
\tau_{j,k}&\sim\frac{1}{2}\text{Gamma}(s_{jk,1},t_{jk,1})+\frac{1}{2}\text{Gamma}(s_{jk,2},t_{jk,2}), \\
\alpha_d &\sim\text{Unif}(k_1,k_2).
\end{align*}
Finally, $\gamma_j \sim\text{IG}(a,b),\ j=1,\ldots,q.$

Define $\gamma=(\gamma_1,\ldots,\gamma_q)^\top$ and $\tau=(\tau_2,\tau_3^\top,\ldots,\tau_q^\top )^\top$. Let $I\lbrace\cdot\rbrace$ denote the indicator function and $\lVert\cdot \rVert_2$ denote the Euclidean norm.  The posterior distribution of this model is characterized by
\begin{align}
& q(B,\Lambda,\alpha_b,\delta,\gamma,\tau,\alpha_d|Y,X)\nonumber\\
\propto\ & \left(\frac{1}{\gamma_1}\right)^{n/2}\exp\left\lbrace-\frac{1}{2\gamma_1}\lVert Y^1-XB^1\rVert_2^2\right\rbrace\nonumber\\
&\qquad\cdot\prod_{j=2}^q\left[\left(\frac{1}{\gamma_j}\right)^{n/2}\exp\left\lbrace-\frac{1}{2\gamma_j}\lVert Y^j-XB^j-(Y^{1:(j-1)}-XB^{1:(j-1)})\delta_j\rVert_2^2\right\rbrace\right]\nonumber\\
&\qquad\cdot\left(\frac{\alpha_b}{2^{1/\alpha_b}\Gamma(1/\alpha_b)}\right)^{pq}\left(\prod_{k=1}^p\prod_{j=1}^q\frac{\lambda_{kj}}{\gamma_j}\right)^{1/\alpha_b}\exp\left\lbrace-\frac{1}{2}\sum_{k=1}^p\sum_{j=1}^q\frac{\lambda_{kj}}{\gamma_j}|B_{kj}|^{\alpha_b}\right\rbrace\nonumber\\
&\qquad\cdot\prod_{k=1}^p\prod_{j=1}^q\left[\frac{f_{kj,1}^{e_{kj,1}}}{\Gamma(e_{kj,1})}\lambda_{kj}^{e_{kj}-1}\exp\left\lbrace-\lambda_{kj}f_{kj,1}\right\rbrace + \frac{f_{kj,2}^{e_{kj,2}}}{\Gamma(e_{kj,2})}\lambda_{kj}^{e_{kj,2}-1}\exp\left\lbrace-\lambda_{kj}f_{kj,2}\right\rbrace\right]\nonumber\\
&\qquad\cdot I\lbrace k_1\le\alpha_b\le k_2\rbrace I\lbrace k_1\le\alpha_d\le k_2\rbrace \prod_{j=1}^q\left(\frac{1}{\gamma_j}\right)^{a+1}\exp\left\lbrace-b/\gamma_j \right\rbrace\nonumber\\
&\qquad\cdot \prod_{j=2}^q\left[\left(\frac{\alpha_d}{2^{1/\alpha_d}\gamma_j^{1/\alpha_d}\Gamma(1/\alpha_d)}\right)^{j-1}\left(\prod_{k=1}^{j-1}\tau_{j,k}\right)^{1/\alpha_d}\exp\left\lbrace-\frac{1}{2\gamma_j}\sum_{k=1}^{j-1}\tau_{j,k}|\delta_{j,k}|^{\alpha_d}\right\rbrace\right.\nonumber\\
&\left.\qquad\quad\cdot\prod_{k=1}^{j-1}\left(\frac{t_{jk,1}^{s_{jk,1}}}{\Gamma(s_{jk,1})}\tau_{j,k}^{s_{jk,1}-1}\exp\left\lbrace-\tau_{j,k}t_{jk,1}\right\rbrace+\frac{t_{jk,2}^{s_{jk,2}}}{\Gamma(s_{jk,2})}\tau_{j,k}^{s_{jk,2}-1}\exp\left\lbrace-\tau_{j,k}t_{jk,2}\right\rbrace\right)\right].\label{eq:posterior}
\end{align}

\subsection{Posterior Conditionals}\label{app:postcond_mh}
The following posterior conditionals are used to construct the GMCB-MH algorithm.

\subsubsection{\texorpdfstring{$B_{kj}$}{bkj}}
Let $B_{(kj)}$ denote the matrix $B$ with $B_{kj}$ removed. The kernel of the posterior conditional for $B_{kj}$ is given by
\begin{align*}
&q(B_{kj} |Y,X, B_{(kj)},\Lambda,\alpha_b,\delta,\gamma,\tau,\alpha_d)\\
\propto&\exp\left\lbrace-\frac{1}{2\gamma_j}\lambda_{kj}|B_{kj}|^{\alpha_b}\right\rbrace\exp\left\lbrace-\frac{1}{2\gamma_1}\lVert Y^1-XB^1\rVert_2^2\right\rbrace\\
&\qquad\cdot\prod_{j=2}^q\exp\left\lbrace-\frac{1}{2\gamma_j}\lVert Y^j-XB^j-(Y^{1:(j-1)}-XB^{1:(j-1)})\delta_j\rVert_2^2\right\rbrace\\
=&\exp\left\lbrace-\frac{1}{2\gamma_j}\lambda_{kj}|B_{kj}|^{\alpha_b}\right\rbrace\exp\left\lbrace-\frac{1}{2}\tr((Y-XB)\Omega(Y-XB)^\top)\right\rbrace,
\end{align*}
where in the last line, $\Omega$ is a function of $\delta$ and $\gamma$ through the modified Cholesky decomposition in equation \eqref{eq:modcholprec}.

\subsubsection{\texorpdfstring{$\lambda_{kj}$}{lambdakj}}
Let $\Lambda_{(kj)}$ denote the matrix $\Lambda$ with $\lambda_{kj}$ removed. The kernel of the posterior conditional for $\lambda_{kj}$ is
\begin{align}
& q(\lambda_{kj}|Y,X,B,\Lambda_{(kj)},\alpha_b,\delta,\gamma,\tau,\alpha_d) \nonumber\\
\propto\ & \frac{f_{kj,1}^{e_{kj,1}}}{\Gamma(e_{kj,1})}\lambda_{kj}^{e_{kj,1}+1/\alpha_b-1}\exp\left\lbrace-\lambda_{kj}\left(f_{kj,1}+\frac{1}{2\gamma_j}|B_{kj}|^{\alpha_b}\right)\right\rbrace\nonumber\\
&\qquad+\frac{f_{kj,2}^{e_{kj,2}}}{\Gamma(e_{kj,2})}\lambda_{kj}^{e_{kj,2}+1/\alpha_b-1}\exp\left\lbrace-\lambda_{kj}\left(f_{kj,2}+\frac{1}{2\gamma_j}|B_{kj}|^{\alpha_b}\right)\right\rbrace.\label{eq:lambdapostcond}
\end{align}
Integrating over the right side of equation \eqref{eq:lambdapostcond} with respect to $\lambda_{kj}$, the normalizing constant is
$$\frac{f_{kj,1}^{e_{kj,1}}}{\Gamma(e_{kj,1})}\frac{\Gamma(e_{kj,1}+1/\alpha_b)}{(f_{kj,1}+\frac{1}{2\gamma_j}|B_{kj}|^{\alpha_b})^{e_{kj,1}+1/\alpha_b}}+\frac{f_{kj,2}^{e_{kj,2}}}{\Gamma(e_{kj,2})}\frac{\Gamma(e_{kj,2}+1/\alpha_b)}{(f_{kj,2}+\frac{1}{2\gamma_j}|B_{kj}|^{\alpha_b})^{e_{kj,2}+1/\alpha_b}}.$$
Let 
\begin{align*}
w_1&=\frac{f_{kj,1}^{e_{kj,1}}}{\Gamma(e_{kj,1})}\frac{\Gamma(e_{kj,1}+1/\alpha_b)}{(f_{kj,1}+\frac{1}{2\gamma_j}|B_{kj}|^{\alpha_b})^{e_{kj,1}+1/\alpha_b}}\\
w_2&=\frac{f_{kj,2}^{e_{kj,2}}}{\Gamma(e_{kj,2})}\frac{\Gamma(e_{kj,2}+1/\alpha_b)}{(f_{kj,2}+\frac{1}{2\gamma_j}|B_{kj}|^{\alpha_b})^{e_{kj,2}+1/\alpha_b}}.
\end{align*}
Then the posterior conditional distribution of $\lambda_{kj}$ is given by
\begin{align*}
& q(\lambda_{kj}|Y,X,B,\Lambda_{(kj)},\alpha_b,\delta,\gamma,\tau,\alpha_d)\\
=\ & \frac{f_{kj,1}^{e_{kj,1}}/\Gamma(e_{kj,1})}{w_1+w_2}\lambda_{kj}^{e_{kj,1}+1/\alpha_b-1}\exp\left\lbrace-\lambda_{kj}\left(f_{kj,1}+\frac{1}{2\gamma_j}|B_{kj}|^{\alpha_b}\right)\right\rbrace\nonumber\\
&\qquad+\frac{f_{kj,2}^{e_{kj,2}}/\Gamma(e_{kj,2})}{w_1+w_2}\lambda_{kj}^{e_{kj,2}+1/\alpha_b-1}\exp\left\lbrace-\lambda_{kj}\left(f_{kj,2}+\frac{1}{2\gamma_j}|B_{kj}|^{\alpha_b}\right)\right\rbrace.
\end{align*} 
Therefore,
\begin{align*}
\lambda_{kj}|Y,X,B,\Lambda_{(kj)},\alpha_b,\delta,\gamma,\tau,\alpha_d&\sim\frac{w_1}{w_1+w_2}\text{Gamma}\left(e_{kj,1}+\frac{1}{\alpha_b},f_{kj,1}+\frac{1}{2\gamma_j}|B_{kj}|^{\alpha_b}\right)\\
&\qquad+\frac{w_2}{w_1+w_2}\text{Gamma}\left(e_{kj,2}+\frac{1}{\alpha_b},f_{kj,2}+\frac{1}{2\gamma_j}|B_{kj}|^{\alpha_b}\right).
\end{align*}

\subsubsection{\texorpdfstring{$\delta_{j,k}$}{deltack}}
Let $\delta_{(j,k)}$ denote the vector $\delta$ with $\delta_{j,k}$ removed. The kernel of the posterior conditional for $\delta_{j,k}$ is given by
\begin{align*}
& q(\delta_{j,k}|Y,X,B,\Lambda,\alpha_b,\delta_{(j,k)},\gamma,\tau,\alpha_d)\\
\propto & \exp\left\lbrace-\frac{1}{2\gamma_j}\tau_{j,k}|\delta_{j,k}|^{\alpha_d}\right\rbrace\exp\left\lbrace-\frac{1}{2\gamma_j}\lVert Y^j-XB^j-(Y^{1:(j-1)}-XB^{1:(j-1)})\delta_j\rVert_2^2\right\rbrace.
\end{align*}

\subsubsection{\texorpdfstring{$\gamma$}{gamma}}\label{app:gammapostcond}
Let $\gamma_{(j)}$ denote the vector $\gamma$ with $\gamma_j$ removed. The kernel for the posterior conditional for $\gamma_1$ is given by
\begin{align*}
& q(\gamma_1|Y,X,B,\Lambda,\alpha_b,\delta,\tau,\gamma_{(1)},\alpha_d)\\
\propto & \left(\frac{1}{\gamma_1}\right)^{\frac{n}{2}+\frac{p}{\alpha_b}+a+1}\exp\left\lbrace-\frac{1}{\gamma_1}\left(\frac{1}{2}\lVert Y^1-XB^1\rVert_2^2+\frac{1}{2}\sum_{k=1}^p\lambda_{k1}|B_{k1}|^{\alpha_b}+b\right)\right\rbrace,
\end{align*}
so 
\begin{align*}
\gamma_1|Y,X,B,\Lambda,\alpha_b,\delta,\tau,\gamma_{(1)},\alpha_d\sim\text{Gamma}\left(\frac{n}{2}+\frac{p}{\alpha_b}+a,\frac{1}{2}\lVert Y^1-XB^1\rVert_2^2+\frac{1}{2}\sum_{k=1}^p\lambda_{k1}|B_{k1}|^{\alpha_b}+b\right).
\end{align*}

The kernel for the posterior conditional for $\gamma_j$, $j=2,\ldots,q$ is given by
\begin{align*}
& q(\gamma_j|Y,X,B,\Lambda,\alpha_b,\delta,\tau,\gamma_{(j)},\alpha_d)\\
\propto & \left(\frac{1}{\gamma_j}\right)^{\frac{n}{2}+\frac{j-1}{\alpha_d}+\frac{p}{\alpha_b}+a+1}\cdot\exp\left\lbrace-\frac{1}{\gamma_j}\left(\frac{1}{2}\lVert Y^j-XB^j-(Y^{1:(j-1)}-XB^{1:(j-1)})\delta_j\rVert_2^2\right.\right.\\
&\qquad\qquad\qquad\qquad\qquad\qquad\qquad\qquad\quad\quad\left.\left.+\frac{1}{2}\sum_{k=1}^{j-1}\tau_{j,k}|\delta_{j,k}|^{\alpha_d}+\frac{1}{2}\sum_{k=1}^p\lambda_{kj}|B_{kj}|^{\alpha_b}+b\right)\right\rbrace,
\end{align*}
so
\begin{align*}
\gamma_j&|Y,X,B,\Lambda,\alpha_b,\delta,\tau,\gamma_{(1)},\alpha_d\\
&\sim\text{Gamma}\left(\frac{n}{2}+\frac{j-1}{\alpha_d}+\frac{p}{\alpha_b}+a,\right.\\
&\left.\qquad\qquad\qquad \frac{1}{2}\lVert Y^j-XB^j-(Y^{1:(j-1)}-XB^{1:(j-1)})\delta_j\rVert_2^2+\frac{1}{2}\sum_{k=1}^{j-1}\tau_{j,k}|\delta_{j,k}|^{\alpha_d}\right.\\
&\left.\qquad\qquad\qquad\qquad\qquad\qquad\qquad\qquad\qquad\qquad\qquad\qquad\qquad+\frac{1}{2}\sum_{k=1}^p\lambda_{kj}|B_{kj}|^{\alpha_b}+b\right).
\end{align*}

\subsubsection{\texorpdfstring{$\tau_{j,k}$}{taujk}}
Let $\tau_{(j,k)}$ denote the vector $\tau$ with $\tau_{j,k}$ removed. The kernel for the posterior conditional for $\tau_{j,k}$ is given by
\begin{align}
& q(\tau_{j,k}|Y,X,B,\Lambda,\alpha_b,\delta,\gamma,\tau_{(j,k)},\alpha_d) \nonumber\\
\propto\ & \frac{t_{jk,1}^{s_{jk,1}}}{\Gamma(s_{jk,1})}\tau_{j,k}^{s_{jk,1}+1/\alpha_d-1}\exp\left\lbrace-\tau_{j,k}\left(t_{jk,1}+\frac{1}{2\gamma_j}|\delta_{j,k}|^{\alpha_d}\right)\right\rbrace\nonumber\\
&\qquad+\frac{t_{jk,2}^{s_{jk,2}}}{\Gamma(s_{jk,2})}\tau_{j,k}^{s_{jk,2}+1/\alpha_d-1}\exp\left\lbrace-\tau_{j,k}\left(t_{jk,2}+\frac{1}{2\gamma_j}|\delta_{j,k}|^{\alpha_d}\right)\right\rbrace.\label{eq:taupostcond}
\end{align}
Integrating over the right side of equation \eqref{eq:taupostcond} with respect to $\tau_{j,k}$, the normalizing constant is
$$\frac{t_{jk,1}^{s_{jk,1}}}{\Gamma(s_{jk,1})}\frac{\Gamma(s_{jk,1}+1/\alpha_d)}{(t_{jk,1}+\frac{1}{2\gamma_j}|\delta_{j,k}|^{\alpha_d})^{s_{jk,1}+1/\alpha_d}}+\frac{t_{jk,2}^{s_{jk,2}}}{\Gamma(s_{jk,2})}\frac{\Gamma(s_{jk,2}+1/\alpha_d)}{(t_{jk,2}+\frac{1}{2\gamma_j}|\delta_{j,k}|^{\alpha_d})^{s_{jk,2}+1/\alpha_d}}.$$
Let
\begin{align*}
d_1&=\frac{t_{jk,1}^{s_{jk,1}}}{\Gamma(s_{jk,1})}\frac{\Gamma(s_{jk,1}+1/\alpha_d)}{(t_{jk,1}+\frac{1}{2\gamma_j}|\delta_{j,k}|^{\alpha_d})^{s_{jk,1}+1/\alpha_d}}\\
d_2&=\frac{t_{jk,2}^{s_{jk,2}}}{\Gamma(s_{jk,2})}\frac{\Gamma(s_{jk,2}+1/\alpha_d)}{(t_{jk,2}+\frac{1}{2\gamma_j}|\delta_{j,k}|^{\alpha_d})^{s_{jk,2}+1/\alpha_d}}.
\end{align*}
Then the posterior conditional distribution of $\tau_{j,k}$ is given by
\begin{align*}
& q(\tau_{j,k}|Y,X,B,\Lambda,\alpha_b,\delta,\gamma,\tau_{(j,k)},\alpha_d)\\
=\ & \frac{t_{jk,1}^{s_{jk,1}}/\Gamma(s_{jk,1})}{d_1+d_2}\tau_{j,k}^{s_{jk,1}+1/\alpha_d-1}\exp\left\lbrace-\tau_{j,k}\left(t_{jk,1}+\frac{1}{2\gamma_j}|\delta_{j,k}|^{\alpha_d}\right)\right\rbrace\nonumber\\
&\qquad+\frac{t_{jk,2}^{s_{jk,2}}/\Gamma(s_{jk,2})}{d_1+d_2}\tau_{j,k}^{s_{jk,2}+1/\alpha_d-1}\exp\left\lbrace-\tau_{j,k}\left(t_{jk,2}+\frac{1}{2\gamma_j}|\delta_{j,k}|^{\alpha_d}\right)\right\rbrace.
\end{align*}
Therefore,
\begin{align*}
\tau_{j,k}|Y,X,B,\Lambda,\alpha_b,\delta,\gamma,\tau_{(j,k)},\alpha_d&\sim\frac{d_1}{d_1+d_2}\text{Gamma}\left(s_{jk,1}+\frac{1}{\alpha_d},t_{jk,1}+\frac{1}{2\gamma_j}|\delta_{j,k}|^{\alpha_d}\right)\\
&\qquad+\frac{d_2}{d_1+d_2}\text{Gamma}\left(s_{jk,2}+\frac{1}{\alpha_d},t_{jk,2}+\frac{1}{2\gamma_j}|\delta_{j,k}|^{\alpha_d}\right).
\end{align*}

\subsubsection{\texorpdfstring{$\alpha_b$}{alphab}}
The kernel of the posterior conditional for $\alpha_b$ is given by
\begin{align*}
& q(\alpha_b|Y,X,B,\Lambda,\delta,\gamma,\tau,\alpha_d)\\
\propto\ & \left(\frac{\alpha_b}{2^{1/\alpha_b}\Gamma(1/\alpha_b)}\right)^{pq}\left(\prod_{k=1}^p\prod_{j=1}^q\frac{\lambda_{kj}}{\gamma_j}\right)^{1/\alpha_b}\exp\left\lbrace-\frac{1}{2}\sum_{k=1}^p\sum_{j=1}^q\frac{\lambda_{kj}}{\gamma_j}|B_{kj}|^{\alpha_b}\right\rbrace I\lbrace k_1\le\alpha_b\le k_2\rbrace.
\end{align*}

\subsubsection{\texorpdfstring{$\alpha_d$}{alphad}}
The kernel of the posterior conditional for $\alpha_d$ is given by
\begin{align*}
&q(\alpha_d|Y,X,B,\Lambda,\alpha_b,\delta,\gamma,\tau)\\
\propto & \prod_{j=2}^q\left(\frac{\alpha_d}{2^{1/\alpha_d}\gamma_j^{1/\alpha_d}\Gamma(1/\alpha_d)}\right)^{j-1}\left(\prod_{k=1}^{j-1}\tau_{j,k}\right)^{1/\alpha_d}\exp\left\lbrace-\frac{1}{2\gamma_j}\sum_{k=1}^{j-1}\tau_{j,k}|\delta_{j,k}|^{\alpha_d}\right\rbrace I\lbrace k_1\le\alpha_d\le k_2\rbrace.
\end{align*}

\subsection{Posterior Conditionals under SMN}\label{app:postcond_smn}
The parameters $B$ and $\delta$ can be easily sampled under a SMN representation. For the exponential power distribution with density proportional to $\exp\lbrace-|x|^\alpha\rbrace$, the mixing density in terms of the standard deviation $\sigma$ is $h(\sigma)\propto\sigma^{-2}p_{\alpha/2}(\sigma^{-2})$, where $p_{\alpha/2}$ is the density of a positive stable distribution with index $\alpha/2<1$ \citep{west1987}. The mixing density in terms of the precision $\omega=1/\sigma^2$ is then $g(\omega)\propto \omega^{-1/2} p_{\alpha/2}(\omega)$.

\subsubsection{\texorpdfstring{$B_{kj}$}{bkj}}
Let $\omega_{kj}$ denote the latent variable associated with $B_{kj}$. Under the SMN representation, the prior on $B_{kj}$ is given by
\begin{align*}
    B_{kj}|\omega_{kj},\lambda_{kj},\gamma_j,\alpha_b&\sim N\left(0,\frac{1}{\omega_{kj}}\bigg(\frac{2\gamma_j}{\lambda_{kj}}\bigg)^{2/\alpha_b}\right),\\
    g(\omega_{kj}|\alpha_b)&\propto\omega_{kj}^{-1/2}\ p_{\alpha_b/2}(\omega_{kj}).
\end{align*}
Define $\omega\in\mathbb{R}^{p\times q}$ to be the matrix of latent variables associated with $B$. Then
\begin{align}
    &q(B,\omega|Y,X,\Lambda,\alpha_b,\delta,\gamma,\tau,\alpha_d) \nonumber\\
    \propto & \exp\left\lbrace-\frac{1}{2}\tr((Y-XB)\Omega(Y-XB)^\top)\right\rbrace\nonumber\\
    &\quad\cdot\prod_{k=1}^p\prod_{j=1}^q\left[\left(\frac{\omega_{kj}\lambda_{kj}^{2/\alpha_b}}{(2\gamma_j)^{2/\alpha_b}}\right)^{1/2}\exp\left\lbrace-\frac{\omega_{kj}\lambda_{kj}^{2/\alpha_b}}{2(2\gamma_j)^{2/\alpha_b}}B_{kj}^2\right\rbrace \omega_{kj}^{-1/2}p_{\alpha_b/2}(\omega_{kj})\right].\label{eq:bomega_smn}
\end{align}
Let $\omega_{(kj)}$ denote the matrix $\omega$ with $\omega_{kj}$ removed. The posterior conditional for $\omega_{kj}$ associated with equation \eqref{eq:bomega_smn} is given by
\begin{align*}
    q(\omega_{kj}|Y,X,B,\Lambda,\alpha_b,\delta,\gamma,\tau,\alpha_d,\omega_{(kj)})\propto\exp\left\lbrace-\frac{\lambda_{kj}^{2/\alpha_b}B_{kj}^2}{2(2\gamma_j)^{2/\alpha_b}}\omega_{kj}\right\rbrace p_{\alpha_b/2}(\omega_{kj}),
\end{align*}
which is an exponentially tilted positive stable distribution. The algorithm described in \cite{devroye2009} and a modified version of the implementation in the R package \texttt{copula} \citep{copula} is used to sample $\omega_{kj}$.

Let $\otimes$ and $\odot$ denote the Kronecker and Hadamard products, respectively. Define $[\gamma]_p$ to be the vector $\gamma$ with each of its elements repeated $p$ times and the $pq\times pq$ matrix $\Delta$ to be  $$\Delta=\diag[\vect(\omega)\odot\vect(\Lambda)^{2/\alpha_b}]\left(\diag[(2[\gamma]_p)^{2/\alpha_b}]\right)^{-1}.$$
Then the posterior conditional for $B$ associated with equation \eqref{eq:bomega_smn} is given by
\begin{align}
    &q(B|Y,X,\Lambda,\alpha_b,\delta,\gamma,\tau,\alpha_d,\omega)\nonumber\\
    \propto&\exp\left\lbrace-\frac{1}{2}\tr((Y-XB)\Omega(Y-XB)^\top)\right\rbrace\prod_{k=1}^p\prod_{j=1}^q\exp\left\lbrace-\frac{\omega_{kj}\lambda_{kj}^{2/\alpha_b}}{2(2\gamma_j)^{2/\alpha_b}}B_{kj}^2\right\rbrace \nonumber \\
    = &\exp\left\lbrace-\frac{1}{2}\tr((Y-XB)\Omega(Y-XB)^\top)\right\rbrace\exp\left\lbrace-\frac{1}{2}\vect(B)^\top\Delta\vect(B)\right\rbrace.\label{eq:bconditional_smn} 
\end{align}

When $p< n$, let $\hat{B}=(X^\top X)^{-1}X^\top Y$. Define $\Theta=\Omega^{-1}\otimes (X^\top X)^{-1}$ and $\Phi=\Theta^{-1}+\Delta$. Then equation \eqref{eq:bconditional_smn} is proportional to
\begin{align*}
    & \exp\left\lbrace-\frac{1}{2}\tr(X^\top X(B-\hat{B})\Omega(B-\hat{B})^\top )\right\rbrace\exp\left\lbrace-\frac{1}{2}\vect(B)^\top \Delta\vect(B)\right\rbrace\\
    = & \exp\left\lbrace-\frac{1}{2}[\vect(B)-\vect(\hat{B})]^\top \Theta^{-1}[\vect(B)-\vect(\hat{B})]\right\rbrace\exp\left\lbrace-\frac{1}{2}\vect(B)^\top \Delta\vect(B)\right\rbrace\\
    \propto&\exp\left\lbrace-\frac{1}{2}\left(\vect(B)^\top \Phi\vect(B)-2\vect(\hat{B})^\top \Theta^{-1}\vect(B)\right)\right\rbrace\\
    = & \exp\left\lbrace-\frac{1}{2}\left(\vect(B)^\top \Phi\vect(B)-2\vect(\hat{B})^\top \Theta^{-1}\Phi^{-1}\Phi\vect(B)\right)\right\rbrace,
\end{align*}
so
\begin{align}
    \vect(B)&|Y,X,\vect(\Lambda),\alpha_b,\delta,\gamma,\tau,\alpha_d,\vect(\omega) \sim N_{pq}\left(\Phi^{-1}\Theta^{-1}\vect(\hat{B}), \Phi^{-1}\right). \label{eq:smn_b_pln}
\end{align}

To show that equation \eqref{eq:bconditional_smn} is a normal distribution even when $(X^\top X)^{-1}$ does not exist requires a variable transformation. Let $X=UCV^\top$ denote the SVD of $X$, where $U\in\mathbb{R}^{n\times n}$ and $V\in\mathbb{R}^{p\times p}$ are orthonormal and $C\in\mathbb{R}^{n\times p}$. Let $r=\rank(X)$. Define
\begin{align*}
    w &= \max(p-r, 0) \\
    z &= \max(n-r, 0),
\end{align*}
and let $\psi\in\mathbb{R}^{r\times r}$ be the diagonal matrix with diagonal elements the positive singular values of $X$. Then 
\begin{equation*}
    C=\begin{pmatrix}
    \psi & 0_{r\times w} \\
    0_{z\times r} & 0_{z\times w}
    \end{pmatrix}.
\end{equation*}
Define $\eta=V^\top BT^\top$. The Jacobian of this transformation is a constant with respect to $\eta$, and equation \eqref{eq:bconditional_smn} can be rewritten as
\begin{align*}
    & q(\eta|Y,X,\Lambda,\alpha_b,\delta,\gamma,\tau,\alpha_d,\omega) \\
    \propto & \exp\left\lbrace-\frac{1}{2}\left[\tr[D^{-1}(- TY^\top UC\eta - \eta^\top C^\top U^\top YT^\top  + \eta^\top C^\top C\eta)] \right.\right.\\
     & \qquad\qquad\qquad \left.\left. +\vect(V\eta(T^{-1})^\top )^\top \Delta\vect(V\eta(T^{-1})^\top )\right]\right\rbrace\\
    = & \exp\left\lbrace-\frac{1}{2}\left[ -2\tr(D^{-1}TY^\top UC\eta) + \tr(D^{-1}\eta^\top C^\top C\eta) \right.\right.\\
    & \qquad\qquad\qquad \left.\left. + \vect(\eta)^\top (T^{-1} \otimes V)^\top  \Delta (T^{-1} \otimes V) \vect(\eta) \right]\right\rbrace\\
    = & \exp\left\lbrace-\frac{1}{2}\left[ -2\vect(C^\top U^\top YT^\top D^{-1})^\top \vect(\eta) + \vect(C\eta D^{-1/2})^\top \vect(C\eta D^{-1/2}) \right.\right. \\
    & \qquad\qquad\qquad \left.\left.+ \vect(\eta)^\top (T^{-1} \otimes V)^\top  \Delta (T^{-1} \otimes V) \vect(\eta) \right]\right\rbrace\\
    = & \exp\left\lbrace-\frac{1}{2}\left[ -2\vect(C^\top U^\top YT^\top D^{-1})^\top \vect(\eta) + \vect(\eta)^\top (D^{-1/2}\otimes C)^\top (D^{-1/2}\otimes C)\vect(\eta) \right.\right. \\
    & \qquad\qquad\qquad \left.\left.+ \vect(\eta)^\top (T^{-1} \otimes V)^\top  \Delta (T^{-1} \otimes V) \vect(\eta) \right]\right\rbrace\\
    = & \exp\left\lbrace-\frac{1}{2}\left[ -2\vect(C^\top U^\top YT^\top D^{-1})^\top \vect(\eta) + \vect(\eta)^\top (D^{-1}\otimes C^\top C)\vect(\eta) \right.\right. \\
    & \qquad\qquad\qquad \left.\left.+ \vect(\eta)^\top (T^{-1} \otimes V)^\top  \Delta (T^{-1} \otimes V) \vect(\eta) \right]\right\rbrace\\
    = & \exp\left\lbrace-\frac{1}{2}\left[ -2\vect(C^\top U^\top YT^\top D^{-1})^\top \vect(\eta) \right.\right.\\
    & \left.\left. \qquad\qquad\qquad + \vect(\eta)^\top [(D^{-1}\otimes C^\top C)+(T^{-1} \otimes V)^\top  \Delta (T^{-1} \otimes V)]\vect(\eta) \right]\right\rbrace.
\end{align*}
Define $\Theta = [(D^{-1}\otimes C^\top C)+(T^{-1} \otimes V)^\top  \Delta (T^{-1} \otimes V)]^{-1}$. Then
\begin{align}
    \vect(\eta)&|Y,X,\vect(\Lambda),\alpha_b,\delta,\gamma,\tau,\alpha_d,\vect(\omega) \sim N_{pq}\left(\Theta\vect(C^\top U^\top YT^\top D^{-1}), \Theta\right).\label{eq:etapostcond_smn}
\end{align}
For $p < n$, $B$ is sampled directly using equation \eqref{eq:smn_b_pln}. For $p \ge n$, $B$ is obtained by first sampling from equation \eqref{eq:etapostcond_smn} using the algorithm by \cite{bhattacharyaetal2016}:
\begin{enumerate}
    \item Sample $u\sim N(0, (T \otimes V^\top ) \Delta^{-1} (T^\top  \otimes V))$ and $e\sim N(0,I_{nq})$ independently.
    \item Set $v=(D^{-1/2}\otimes C)u + e$.
    \item Solve for $w$ in 
    $$[(D^{-1/2}T \otimes CV^\top )\Delta^{-1}(T^\top D^{-1/2}\otimes VC^\top ) + I_{nq}]w=\vect(U^\top YT^\top D^{-1/2}) - v.$$
    \item Set $\vect(\eta) = u + (T \otimes V^\top )\Delta^{-1}(T^\top D^{-1/2}\otimes VC^\top )w.$
\end{enumerate}

\subsubsection{\texorpdfstring{$\delta_j$}{deltaj}}
Let $\epsilon_{j,k}$ be the latent variable associated with $\delta_{j,k}$. The prior on $\delta_{j,k}$ can be represented as
\begin{align*}
    \delta_{j,k}|\epsilon_{j,k},\tau_{j,k},\gamma_j,\alpha_d&\sim N\left(0,\frac{1}{\epsilon_{j,k}}\bigg(\frac{2\gamma_j}{\tau_{j,k}}\bigg)^{2/\alpha_d}\right),\\
    g(\epsilon_{j,k}|\alpha_d)&\propto\epsilon_{j,k}^{-1/2}\ p_{\alpha_d/2}(\epsilon_{j,k}).
\end{align*}
Define $\epsilon\in\mathbb{R}^{q(q-1)/2}$ to be the vector of latent variables associated with $\delta$. Then
\begin{align}
    & q(\delta,\epsilon|Y,X,B,\Lambda,\alpha_b,\gamma,\tau,\alpha_d)\nonumber\\
    \propto & \prod_{j=2}^q\left[\exp\left\lbrace-\frac{1}{2\gamma_j}\lVert Y^j-XB^j-(Y^{1:(j-1)}-XB^{1:(j-1)})\delta_j\rVert_2^2\right\rbrace\right.\nonumber\\
    &\left.\qquad\quad\cdot\prod_{k=1}^{j-1}\left(\frac{\epsilon_{j,k}\tau_{j,k}^{2/\alpha_d}}{(2\gamma_j)^{2/\alpha_d}}\right)^{1/2}\exp\left\lbrace-\frac{\epsilon_{j,k}\tau_{j,k}^{2/\alpha_d}}{2(2\gamma_j)^{2/\alpha_d}}\delta_{j,k}^2\right\rbrace\epsilon_{j,k}^{-1/2}\ p_{\alpha_d/2}(\epsilon_{j,k})\right].\label{eq:deltaepsilon_smn}
\end{align}
Let $\epsilon_{(j,k)}$ denote the vector $\epsilon$ with $\epsilon_{j,k}$ removed. The posterior conditional for $\epsilon_{j,k}$ corresponding to equation \eqref{eq:deltaepsilon_smn} is
\begin{align*}
    q(\epsilon_{j,k}|Y,X,B,\Lambda,\alpha_b,\delta,\gamma,\tau,\alpha_d,\epsilon_{(j,k)}) \propto\exp\left\lbrace-\frac{\tau_{j,k}^{2/\alpha_d}\delta_{j,k}^2}{2(2\gamma_j)^{2/\alpha_d}}\epsilon_{j,k}\right\rbrace p_{\alpha_d/2}(\epsilon_{j,k}).
\end{align*}
This is an exponentially tilted positive stable distribution and can be sampled from using the algorithm described in \cite{devroye2009}.

Let $Z_j=Y^j-XB^j$, $W_j=Y^{1:(j-1)}-XB^{1:(j-1)}$, $\tau_j=(\tau_{j,1},\ldots,\tau_{j,j-1})^\top $, and $\epsilon_j=(\epsilon_{j,1},\ldots,\epsilon_{j,j-1})^\top $. Define the $(j-1)\times(j-1)$ matrix $\Psi_j$ to be $\Psi_j=\diag[(\epsilon_j\odot\tau_j^{2/\alpha_d})/(2\gamma_j)^{2/\alpha_d}]$, and let $\delta_{(j)}$ denote the vector $\delta$ with $\delta_j$ removed. Then the posterior conditional for $\delta_j$ corresponding to equation \eqref{eq:deltaepsilon_smn} is
\begin{align*}
    &q(\delta_j|Y,X,B,\Lambda,\alpha_b,\delta_{(j)},\gamma,\tau,\alpha_d,\epsilon)\\
    \propto&\exp\left\lbrace-\frac{1}{2\gamma_j}\lVert Z_j-W_j\delta_j\rVert_2^2\right\rbrace\exp\left\lbrace-\frac{1}{2}\delta_j^\top \Psi_j\delta_j\right\rbrace\\
    \propto&\exp\left\lbrace-\frac{1}{2\gamma_j}\left(\delta_j^\top W_j^\top W_j\delta_j-2Z_j^\top W_j\delta_j+\delta_j^\top (\gamma_j\Psi_j)\delta_j\right)\right\rbrace\\
    =&\exp\left\lbrace-\frac{1}{2\gamma_j}\left(\delta_j^\top (W_j^\top W_j+\gamma_j\Psi_j)\delta_j-2Z_j^\top W_j(W_j^\top W_j+\gamma_j\Psi_j)^{-1}(W_j^\top W_j+\gamma_j\Psi_j)\delta_j\right)\right\rbrace,
\end{align*}
so
\begin{align*}
    \delta_j|Y,X,B,\Lambda,\alpha_b,\delta_{(j)},\gamma,\tau,\alpha_d,\epsilon\sim N_{j-1}\left((W_j^\top W_j+\gamma_j\Psi_j)^{-1}W_j^\top Z_j,\gamma_j(W_j^\top W_j+\gamma_j\Psi_j)^{-1}\right).
\end{align*}
Recall that $j=2,\ldots,q$. For $j \le n$, $\delta_j$ is sampled directly.  For $j> n$, the algorithm proposed by \cite{bhattacharyaetal2016} is used:
\begin{enumerate}
    \item Sample $u\sim N(0,\Psi_j^{-1})$ and $e\sim N(0,I_n)$ independently.
    \item Set $v=\frac{1}{\sqrt{\gamma_j}}W_j u + e$.
    \item Solve for $w$ in $(\frac{1}{\gamma_j}W_j\Psi_j^{-1}W_j^\top +I_n)w=\frac{1}{\sqrt{\gamma_j}}Z_j-v$.
    \item Set $\delta_j=u+\frac{1}{\sqrt{\gamma_j}}\Psi_j^{-1}W_j^\top w$.
\end{enumerate}

\section{Additional autocorrelation plots}
Figures \ref{fig:acfcomparison_2} and \ref{fig:acfcomparison_3} display the autocorrelation plots for randomly selected elements of $B$ and $\delta$ in Scenarios 2 and 3 as described in \ref{sec:sim}. These plots are typical of what was observed for $B$ and $\delta$ in each scenario.

\begin{figure}
    \centering
    \includegraphics[width=0.7\textwidth]{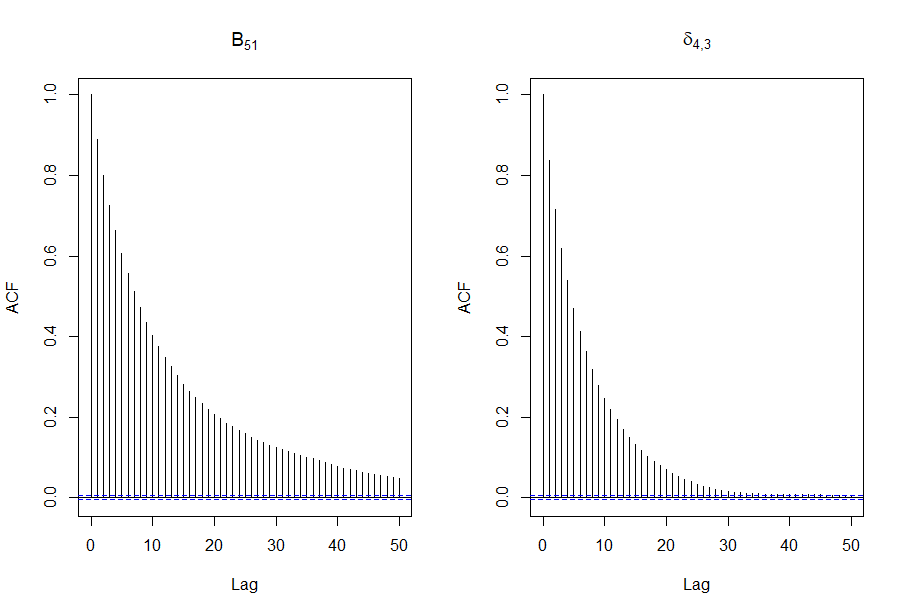}
    \includegraphics[width=0.7\textwidth]{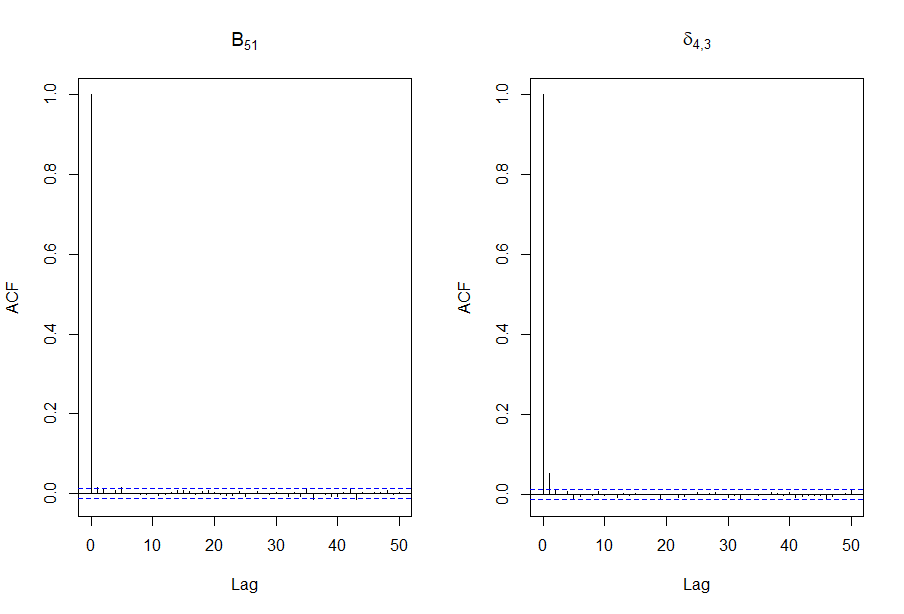}
    \caption{Autocorrelation plots for randomly selected elements of $B$ and $\delta$ in Scenario 2 in Section \ref{sec:sim}. The top plots correspond to GMCB-MH and the bottom plots correspond to GMCB-SMN.}
    \label{fig:acfcomparison_2}
\end{figure}
\FloatBarrier

\begin{figure}
    \centering
    \includegraphics[width=0.7\textwidth]{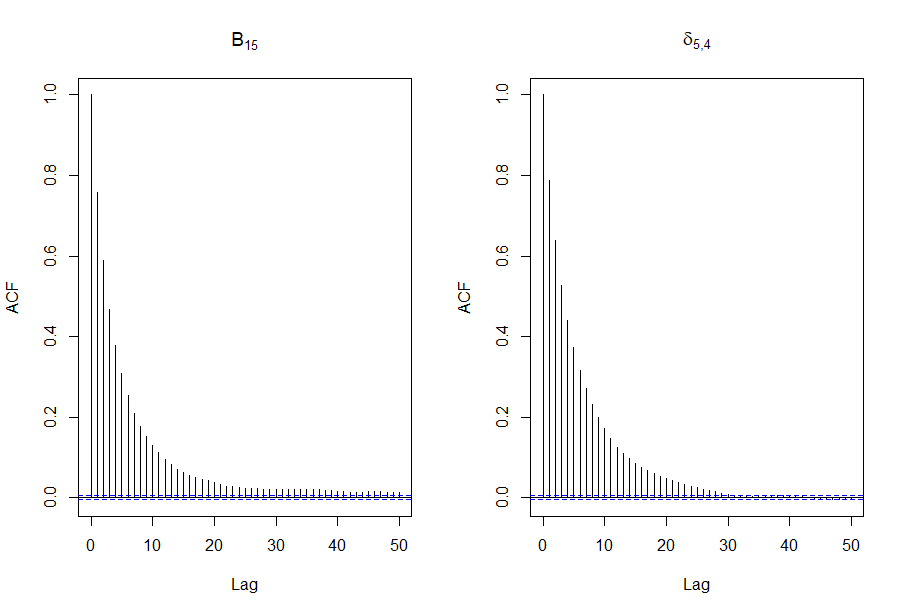}
    \includegraphics[width=0.7\textwidth]{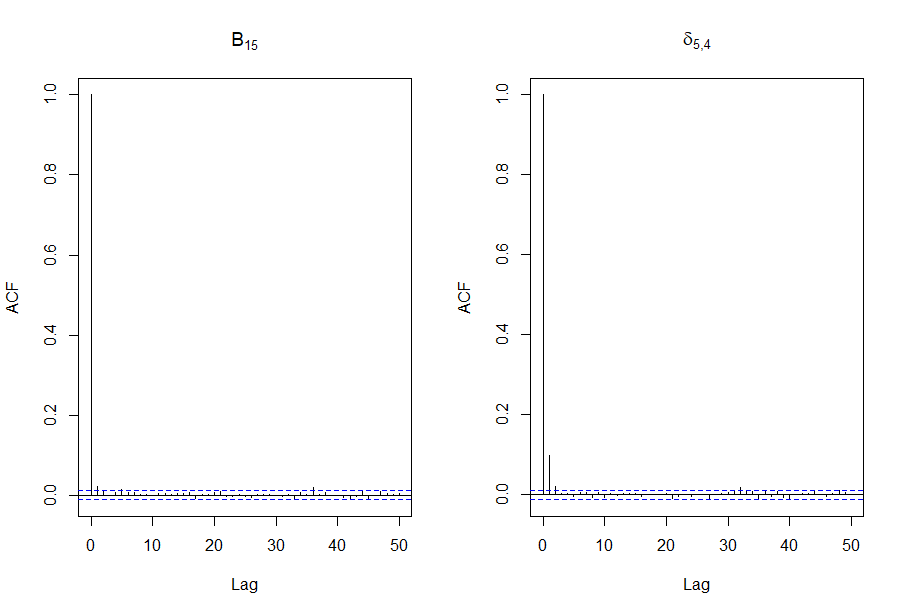}
    \caption{Autocorrelation plots for randomly selected elements of $B$ and $\delta$ in Scenario 3 in Section \ref{sec:sim}. The top plots correspond to GMCB-MH and the bottom plots correspond to GMCB-SMN.}
    \label{fig:acfcomparison_3}
\end{figure}
\FloatBarrier

\section{Posterior Credible Interval Coverage}\label{app:postci}

Table \ref{table:postcicoverage} give the average nominal coverage of the 95\% posterior credible intervals produced by GMCB for Scenarios 1--3 in Section \ref{sec:sim}. 

\begin{table}[!ht]
    \centering
    \caption{Average nominal coverage of 95\% posterior credible intervals computed from 2000 replications and averaged over the entries of $B$ and the entries of $\Omega$.}
    \begin{tabular}{cccc}
    \hline\rule{0pt}{2.5ex}  
    Scenario & Method & Parameter & Nominal Coverage \\ \hline\rule{0pt}{3ex}
    \multirow{4}{*}{1} & \multirow{2}{*}{GMCB-MH} & $B$ & 0.909 \\
    & & $\Omega$ & 0.948 \\ \cline{2-4}\rule{0pt}{3ex}
    & \multirow{2}{*}{GMCB-SMN} & $B$ & 0.910  \\
    & & $\Omega$ & 0.916 \\ \hline
     \multirow{4}{*}{2} & \multirow{2}{*}{GMCB-MH} & $B$ & 0.704 \\
    & & $\Omega$ & 0.963 \\ \cline{2-4}\rule{0pt}{3ex}
    & \multirow{2}{*}{GMCB-SMN} & $B$ & 0.673  \\
    & & $\Omega$ & 0.969 \\ \hline
     \multirow{4}{*}{3} & \multirow{2}{*}{GMCB-MH} & $B$ & 0.917 \\
    & & $\Omega$ & 0.949 \\ \cline{2-4}\rule{0pt}{3ex}
    & \multirow{2}{*}{GMCB-SMN} & $B$ & 0.916  \\
    & & $\Omega$ & 0.934 \\ \hline
    \end{tabular}
    \label{table:postcicoverage}
\end{table}

The posterior credible intervals for $\Omega$ generally have close-to-nominal coverage. However, the nominal coverage for $B$ is generally closer to about 91\%. When $B$ has zero values, the posterior credible intervals produced by GMCB have higher-than-nominal coverage for the true zero values in $B$, but coverage for the true non-zero values tends to be low. The effect of this difference is particularly apparent in Scenario 2, where the number of zeros and non-zero values is approximately the same. Although Scenario 3 also has a sparse $B$, there are only 3 non-zero values, for which one of the credible intervals has close-to-nominal coverage.

\end{appendix}

\end{document}